\documentclass[smallextended,natbib]{svjour3}
\smartqed

\usepackage{mathtools,amssymb,amsfonts,dsfont,mathrsfs}
\usepackage{algorithmic}
\usepackage{graphicx,subcaption}
\usepackage{xcolor}

\usepackage{hyperref}

\newcommand\reals{\mathds{R}}

\DeclareMathOperator{\BR}{\mathsf{BR}}

\newcommand\ALPHABET{\mathcal}
\DeclareMathOperator\Span{span}
\DeclareMathOperator\Lipschitz{Lip}
\newcommand\Minkowski{\rho}

\newcommand\PR{\mathds{P}}
\newcommand\EXP{\mathds{E}}
\newcommand\IND{\mathds{1}}
\newcommand\ONES{\mathbf{1}}

\DeclareMathOperator\Supp{Supp}

\newcommand{\MDP}{\mathcal M}
\newcommand{\approxMDP}{\widehat{\mathcal{M}}}
\newcommand{\St}{S} 
\newcommand{\st}{s} 
\newcommand{\StSp}{\ALPHABET{\St}} 
\newcommand{\Act}{A}
\newcommand{\act}{a}
\newcommand{\ActSp}{\ALPHABET{\Act}}
\newcommand{\ActSpAll}{(\ALPHABET{\Act}^i)_{i \in \playerset}}
\newcommand{\Trans}{\mathsf{P}}
\newcommand{\approxTrans}{\widehat{\mathsf{P}}}
\newcommand{\arbitraryTrans}{\bar{\mathsf{P}}}

\newcommand{\rew}{r} 
\newcommand{\approxrew}{\hat{r}}
\newcommand*\rewAll{(r^i)_{i \in \playerset}} 
\newcommand*\approxrewAll{(\hat{r}^i)_{i \in \playerset}}

\newcommand\discount{\gamma}
\newcommand{\Pol}{\Pi}
\newcommand{\pol}{\pi}
\newcommand{\approxpol}{\hat{\pi}}
\newcommand{\polAll}{(\pol^i)_{i \in \playerset}}
\newcommand{\polAllApprox}{(\hat \pol^i)_{i \in \playerset}}

\newcommand{\Bellman}{\ALPHABET{B}}
\newcommand{\valf}{v}
\newcommand{\Valf}{V}
\newcommand{\approxValf}{\hat{V}}
\newcommand{\PSP}{\mathscr{P}}
\newcommand{\playerset}{\ALPHABET{N}}

\newcommand{\game}{\mathscr{G}}
\newcommand{\approxgame}{\widehat{\mathscr{G}}}

\newcommand{\F}{\mathfrak{F}}
\newcommand{\FTV}{\mathfrak{F}^{\mathrm{TV}}}
\newcommand{\FW}{\mathfrak{F}^{\mathrm{W}}}
\newcommand{\W}{d_{\FW}}
\newcommand{\TV}{d_{\FTV}}

\newcommand*\rewerror{\varepsilon}
\newcommand*\proberror{\delta}
\newcommand*\aiserror{(\rewerror, \proberror)}

\newcommand*\GAMMAN{0.5824}
\newcommand*\NH{11,122,695}
\newcommand*\NB{833,348}

\begin{document}

\titlerunning{Model-based MARL for general sum Markvov games}
 \title{Robustness and sample complexity of model-based MARL for general-sum Markov games}

\author{Jayakumar Subramanian \and Amit Sinha \and Aditya Mahajan%
\thanks{The work of Amit Sinha and Aditya Mahajan was supported in part by the
Innovation for Defence Excellence and Security (IDEaS) Program of the Canadian
Department of National Defence through grant CFPMN2-30.}%
\thanks{A preliminary version of this work appeared in the 2021 Indian Control Conference.}
}
\institute{J. Subramanian \at Media and Data Science Research Lab, Digital Experience Cloud,\\ Adobe Inc., Noida, Uttar Pradesh, India. \\\email{jasubram@adobe.com} 
  \and A. Sinha and A. Mahajan \at Department of Electrical and Computer
  Engineering,\\ McGill University, Montreal,
Canada.\\\email{amit.sinha@mail.mcgill.ca, aditya.mahajan@mcgill.ca}}

\maketitle

\begin{abstract}
  Multi-agent reinforcement learning (MARL) is often modeled using the framework of Markov games (also called stochastic games or dynamic games). Most of the existing literature on MARL concentrates on zero-sum Markov games but is not applicable to general-sum Markov games. It is known that the best-response dynamics in general-sum Markov games are not a contraction. Therefore, different equilibria in general-sum Markov games can have different values. Moreover, the Q-function is not sufficient to completely characterize the equilibrium. Given these challenges, model based learning is an attractive approach for MARL in general-sum Markov games.

In this paper, we investigate the fundamental question of \emph{sample complexity}  for model-based MARL algorithms in general-sum Markov games. We show two results. We first use Hoeffding inequality based bounds to show that  $\tilde{\mathcal{O}}( (1-\discount)^{-4} \alpha^{-2})$ samples per state-action pair are sufficient to obtain a $\alpha$-approximate Markov perfect equilibrium with high probability,
where 
$\discount$ is the discount factor, and the $\tilde{\mathcal{O}}(\cdot)$ notation hides logarithmic terms.
We then use Bernstein inequality based bounds to show that $\tilde{\mathcal{O}}( (1-\discount)^{-1} \alpha^{-2} )$ samples are sufficient.
To obtain these results, we study the robustness of Markov perfect equilibrium to model approximations. We show that
the Markov perfect equilibrium of an approximate (or perturbed) game is always an approximate Markov perfect equilibrium of the original game and provide explicit bounds on the approximation error. 
We illustrate the results via a numerical example.
\end{abstract}

\section{Introduction}
Markov games (also called stochastic games or dynamic games) are a commonly used framework to model strategic interaction between multiple players interacting in a dynamic environment. Examples include applications in cyber-security~\citep{sengupta2019general}, industrial organization~\citep{ericson1995markov,fershtiman2000dynamic}, political economics~\citep{acemoglu2001theory}, advertisement and pricing~\cite{albright1979birth}, 
and many others~\citep{Basar2018}. Starting with the seminal work of \citet{Shapley_1953}, several variations of Markov games have been considered in the literature.  We refer the reader to~\cite{FilarCompetitive1996} for an overview. 

\smallskip\noindent\textbf{Overview of Markov games.} In the basic setup of a dynamic game, the payoffs of players at any time not only depend on their current joint action profile but also on the current ``state of the system''. Furthermore, the state of the system evolves in a controlled Markov manner conditioned on the current action profile of the players. It is typically assumed that the state of the system and the action profile of all players are publicly monitored by all players. Although Markov games may be viewed as a special case of extensive form games with perfect information, rather than using the standard solution concept of sub-game perfect equilibrium, attention is often restricted to a refinement of sub-game perfect equilibrium called Markov perfect equilibrium (MPE) where all players play Markov strategies (i.e., choose their actions as a (possibly randomized) function of the current state)~\citep{maskin1988theory, maskin1988theory2}. MPE is an attractive refinement of sub-game perfect equilibrium, both from a computational as well as conceptual point of view, but has some limitations because it excludes some history dependent strategies (such as tit-for-tat and grim trigger) commonly used in the repeated games setup. See~\cite{maskin2001markov, MailathRepeated} for a discussion.

Games can also be classified based on the sum of per-step payoffs of players as zero-sum or general-sum games. The nature of results in these two cases are different as are the tools used to prove them. The differences stem from the fact that the best response mapping (called the Shapley operator) for two-player zero-sum games is a contraction~\citep{Shapley_1953}. Therefore, zero-sum games have a unique value (i.e., all equilibria in zero-sum games have the same value). Moreover, the MPE (also called minimax equilibrium for the zero-sum case) can be computed via recursive operations of the Shapley operator~\citep{Shapley_1953,hoffman1966nonterminating}. In contrast, the best response mapping for general-sum games is not a contraction. Therefore, the existence of MPE needs to be proved using variations of Kakutani's fixed point theorem \citep{fink1964equilibrium, TakahashiEquilibrium1964, RogersNonzerosum1969, VriezeStochastic1987}. A consequence of this is that, in general, different MPEs do not have the same value, which makes it difficult to compute MPEs. Various algorithms have been proposed to compute MPE, including non-linear programming \citep{Breton1991,filar1991nonlinear} and homotopy methods \citep{herings2004stationary, herings2010homotopy}. It was recently established by \cite{deng2021complexity} that the computational complexity of computing MPE is PPAD-complete.

In spite of these challenges, computing MPE of general-sum games is an important research direction  because several real-world problems are not zero-sum. The applications of network security~\citep{sengupta2019general}, industrial organization~\citep{ericson1995markov,fershtiman2000dynamic}, and political economics~\citep{acemoglu2001theory} mentioned above are all general-sum games.

\smallskip\noindent\textbf{Multi-agent reinforcement learning.}
In recent years, there has been significant interest in understanding interaction between strategic agents operating in unknown environments. Such multi-player problems are studied under the heading of multi-agent reinforcement learning (MARL) and often modeled as Markov games \citep{littman1994markov,busoniu2008comprehensive,zhang2021multi}. Although there have been significant recent successes in single agent RL, these do not directly translate into the multi-agent setting. Part of the difficulty is that when multiple agents are learning simultaneously, the ``environment'' as viewed by any single agent is non-stationary \citep{busoniu2008comprehensive}; so it is not possible to use the theoretical guarantees of single agent RL algorithms, which are derived for a stationary or time-homogeneous environment. 

Nonetheless, MARL for two player zero-sum games is well understood due to two properties.
First, if two strategies $(\pi^1, \pi^2)$ and $(\mu^1, \mu^2)$ are minimax
equilibrium, then so are strategies $(\pi^1, \mu^2)$ and $(\mu^1,
\pi^2)$. Therefore, to identify equilibrium strategies, it is sufficient to
learn the action-value function (i.e., the Q-function). Second, the
action-value function can be learnt using variants of Q-learning (called
minimax Q-learning) because the Shapley operator is a contraction
\citep{littman1994markov, littman2001value}. 
We refer the reader to \cite{shoham2003multi} for an overview of MARL for zero-sum games.

However, the situation is different for general-sum MARL, where fewer
convergence guarantees are available. Part of the difficulty is that the
action-value function (or $Q$-function) is insufficient to characterize MPE
\citep[Theorem 1]{zinkevich2006cyclic}.\footnote{\citet{zinkevich2006cyclic}
  construct two player general-sum games with the following properties. The
  game has two states: in state~1, player~1 has two actions and player~2 has
  one action; in state~2, player~1 has one action and player~2 has two
  actions. The transition probabilities are chosen such that there is a unique
  Markov perfect equilibrium in mixed strategies. This means that in state~1, both
  actions of player~1 maximize the $Q$-function; in state~2, both actions of
  player~2 minimize the $Q$-function. However, the $Q$-function in itself is
  insufficient to determine the randomizing probabilities for the mixed
strategy MPE.} For this reason, algorithms 
developed for two-player zero-sum games fail to converge to an MPE in
general-sum games~\citep{perolat2017learning}. There are some partial results,
e.g., minimizing Bellman residual error to identify $\varepsilon$-MPE
\citep{perolat2017learning}, using two-time scale stochastic approximation
algorithms \citep{prasad2015two}, and using replicator dynamics based
algorithms \citep{akchurina2010multi}. However, in general, developing MARL
algorithms with convergence guarantees remains a challenging research
direction.

There is some recent work on learning in general-sum games.
    \cite{leonardos2021global} show global convergence of a policy gradient algorithm in a special class of Markov games called Markov potential games, which are a generalization of normal form potential games and assume the existence of a common potential function for all players. \cite{zhang2021gradient} also consider Markov potential games and present convergence analysis for a sample-based RL method in such games. While both these results are interesting, they do not apply to Markov games which do not have a potential function. Another recent result is presened by \cite{song2021can}, who present
    present and analyse algorithms for computing correlated and coarse correlated equilibria for general-sum games. The solution concepts of correlated equilibrium and its variations are different from MPE. In correlated equilibrium, players agree on a joint randomization strategies before the system starts running; such pre-game agreement is not allowed in MPE.

\smallskip\noindent\textbf{Model based MARL, sample complexity, and robustness of equilibria.}
One potential approach to alleviate the difficulties in MARL for general-sum games is to use model based algorithms, which explicitly learn (or estimate) the system model and then use a ``planning algorithm'' to find the solution of the estimated model \citep{sutton1990}. There has been significant recent interest in model based RL for single agent systems (see \cite{Wang2019Benchmarking} and references therein) and some interest in model-based approaches for MARL for zero-sum games \citep{Krupnik2019MultiAgent,sidford2020solving,Zhang2021ModelBased,zhang2020model}. However, as far as we are aware, there are no model based MARL algorithms for general-sum Markov games.

An important consideration in model-based RL is to determine how many samples are needed to identify an $\alpha$-approximate solution (for a pre-specified accuracy level $\alpha$). This is known as \emph{sample complexity} of learning and is typically analyzed under the assumption that the learning agent has access to a generative model, i.e., a black box simulator that takes the current state and action profile as input and generates samples of the next state as output.

Starting with the work of~\cite{kearns1999finite,Kakade:PhD}, there is an extensive literature on the sample complexity of Markov decision processes (MDPs) \citep{azar2013minimax,sidford2018near,agarwal2020model,li2020breaking}. The simplest approach in this setting is to use a plug-in estimator,\footnote{The plug-in estimator is also known as a certainty equivalent controller in the stochastic control literature.} i.e., estimating the transition matrix using the generated samples and using the optimal policy corresponding to the estimated model in the true system. Recent results of \cite{agarwal2020model} show that the sample complexity of the plug-in estimator matches the lower bounds on sample complexity~\citep{azar2013minimax} modulo logarithmic factors. Recently, \cite{zhang2020model}, build on this line of work to establish sample complexity bounds for zero-sum games. As far as we are aware, sample complexity of generative models for general-sum games hasn't been investigated before.

The analyses of model-based RL algorithms rely on the \emph{robustness} of the ``planning solution'' to model approximations, i.e., \emph{if the estimated model is close to the true model in some sense, does that imply that the strategy generated from the estimated model is approximately appropriate in some sense (optimality, equilibrium, etc.)?} This question is well understood for Markov decision processes (see \cite{Muller1997} and follow-up work) and zero-sum Markov games \citep{Tidball1996Approximations, Tidball1997Approximations}. In this paper, we address the question of robustness for general-sum Markov games. In particular, we show that if a dynamic game is approximated by another game such that the reward functions and transitions of the approximate game are close to those of the original game (in an appropriate sense), then a MPE of the approximate game is an approximate MPE of the original game. We quantify the exact relationship between the degree of approximation of the games and the approximation error in the MPE. 
We then build up on these results to establish sample complexity bounds for learning with a generative model for general-sum Markov games. 

The notion of robustness is also useful in its own right. In many applications, the model of a dynamic game is estimated using modern econometric techniques~\cite{aguirregabiria2007sequential, bajari2007estimating, pakes2007simple,pesendorfer2008asymptotic}.  In such situations, robustness characterizes the approximation error in using a MPE of an approximate game, in terms of the approximation errors in estimating the reward function and transition dynamics of the game.

\smallskip\noindent\textbf{Other notions of robustness.}
Our notion of robustness is different from that of robust control \citep{Basar2008} and robust Markov perfect equilibrium \citep{Jaskiewicz2014}, both of which are Markov decision processes with uncertain dynamics and are treated as zero-sum games where nature acts as an adversary and picks the worst-case realization of the transition dynamics. Our notion of robustness is also different from
uniformly $\varepsilon$-equilibrium \citep{Solan2021}, which captures robustness with respect to time-horizon and discount factor.

Our notion of robustness is similar in spirit to robust MPEs considered in \cite{maskin2001markov}, who defined a MPE to be robust if for any small perturbation of the payoffs, there exists a nearby MPE. \cite{maskin2001markov} showed that almost all finite horizon general-sum games have a finite number of MPEs, all of which are robust. Our results are of a different nature and it is difficult to compare the two results because \cite{maskin2001markov} considered an atypical model where the states are not specified exogenously but are rather determined as the payoff relevant component of the history. Consequently, perturbing the payoffs changes the state \emph{space}, which is not the case for our model. 

The notion of strong stability considered in \cite{Doraszelski2010} is related to the work in \cite{maskin2001markov}. It is shown in \cite{Doraszelski2010} that almost all Markov games have a finite number of MPEs and these equilibria can be approximated by equilibria of nearby games.
The dynamics in \cite{Doraszelski2010} are exogenous and, therefore, their result does not have the same limitations as that of \cite{maskin2001markov}. The result of \cite{Doraszelski2010} is stronger than ours because we only show that equilibria of nearby games are approximate equilibria of the original game but we do not establish that they are also close to the equilibria of the original game. However, the results of \cite{Doraszelski2010} rely on continuity arguments and do not explicitly characterize bounds on the size of the neighborhood. In contrast, for any $\varepsilon$ perturbation in payoffs and $\delta$ perturbation in dynamics, we explicitly characterize an $\alpha$ such that the MPE of the perturbed game is an $\alpha$-MPE of the original game.

Perhaps the result most similar to ours is \cite{whitt1980representation}, who consider a more general model and allow the approximate game to have a different state and action space than the original game. Their main result is to show that any $\alpha_{\mathrm{opt}}$-MPE of the approximate game is an $\alpha$-MPE of the original game and an explicit relationship between $\alpha_{\mathrm{opt}}$ and $\alpha$ is established. Our results are similar in spirit but the specific details are different. 

\smallskip \noindent \textbf{Organization.}
The rest of the paper is organized as follows. In Sec.~\ref{sec:model}, we present our notion of approximation of a dynamic game and state our main results. In Sec.~\ref{sec:MDP}, we present background results on approximation of Markov decision processes. In Sec.~\ref{sec:proof}, we provide the proof of our main results. In Sec.~\ref{sec:numerical_examples}, we present numerical examples to validate our theoretical results. We conclude in Sec.~\ref{sec:conclusion}.

\smallskip \noindent \textbf{Notation.} 
We use $\reals$ to denote the set of real numbers, $\PR(\cdot)$ to denote the probability of an event, $\EXP[\cdot]$ to denote the expectation of a random variable, and $\PSP(\cdot)$ denotes the set of probability measures on a set.

We use calligraphic letters (e.g., $\StSp$, $\ActSp$, etc.) to denote sets, uppercase letters (e.g., $\St$, $\Act$, etc.) to denote random variables and lowercase letters (e.g., $\st$, $\act$, etc.) to denote their realization. 
Superscripts index players and subscripts index time. For example, $\act^i_t$ denotes the action of player~$i$ at time~$t$. For sequence of variables $\{\st_t\}_{t \ge 1}$, we use the short hand notation $\st_{1:t}$ to denote the sequence $(\st_1, \dots, \st_t)$. We use $\ONES$ to denote a vector of ones of an appropriate size which is determined by context.

Given a function $f \colon \StSp \to \reals$, we use $\Span(f)$ to denote the span seminorm of~$f$, i.e., $\Span(f) \coloneqq \sup_{\st \in \St} f(\st) - \inf_{\st \in \St} f(\st)$. Given a metric space $(\StSp, d)$ and a function $f \colon \StSp \to \reals$, we use $\Lipschitz(f)$ to denote the Lipschitz constant of $f$, i.e., 
\[
  \Lipschitz(f) \coloneqq \sup_{\st, \st' \in \StSp} \frac{ | f(\st) - f(\st') | }{ d(\st, \st') }.
\]

\section{System model, robustness, and sample complexity }\label{sec:model}

We restrict the discussion in this paper to models with finite state and action spaces. The robustness results can be extended to models with continuous state and action spaces under standard technical assumptions on the existence of equilibria in that setting.

\subsection{Dynamic games}
An infinite horizon dynamic game (also called stochastic game or Markov game) is a tuple $\langle \playerset, \StSp, \ActSpAll, \Trans, \rewAll, \discount \rangle$ where:
\begin{itemize}
    \item $\playerset$ is the (finite) set of players.
    \item $\StSp$ is the (finite) set of possible states of the game. We use $\St_t \in \StSp$ to denote the state of the game at time~$t$.
    \item $\ActSpAll$ is the (finite) set of actions available to player $i$ at each time. We also use $\ActSp = \prod_{i \in \playerset}\ActSp^i$ to denote the set of actions of all players. We use $\Act_t = (\Act^i_t)_{i \in \playerset}$ to denote the action profile of all players at time~$t$. Given an action profile $\Act_t = (\Act^i_t)_{i \in \playerset}$ and a player $j \in \playerset$, we use the notation $\Act^{-j}_t = (\Act^i_t)_{i \in \playerset \setminus \{j\}}$ to denote the action profile of all players except $j$.
    \item $\Trans: \StSp \times \ActSp \to \PSP(\StSp)$ is the controlled transition probability of the state of the game. In particular, at any time~$t$, given a realization $\st_{1:t+1}$ of $\St_{1:t+1}$ and choice of action profile $\act_{1:t}$ of $\Act_{1:t}$, we have
    \[
      \Trans(\st_{t+1} \mid \st_t, \act_t) \coloneqq \PR(\St_{t+1} = \st_{t+1} \mid \St_{1:t} = \st_{1:t}, \Act_{1:t} = \act_{1:t}).
    \]
    \item $\rew^i: \StSp \times \ActSp \to \reals$ denotes the per-step reward of player~$i$. 
    \item $\discount \in (0,1)$ is the discount factor.
\end{itemize}

We assume that all players have perfect monitoring. At time~$t$, all players observe the current state $\St_t$ and simultaneously choose their respective actions. At the end of time period $t$, all players observe all the actions, and the state of the game evolves according to the transition kernel $\Trans$.

Following~\cite{Shapley_1953,maskin1988theory,maskin1988theory2}, we assume that each player chooses its action according to a time homogeneous Markov strategy. Let
\begin{equation*}
    \Pol^i \coloneqq \{\pol^i: \StSp \to \PSP(\ActSp^i) \}
\end{equation*}
denote the set of all Markov strategies for player~$i$.

Given a strategy profile $\pol = \polAll$, where $\pol^i \in \Pol^i$, and an initial state $\st_0$, the expected discounted total reward of player $i$ is given by:
\begin{equation}
  V^i_{(\pol^i, \pol^{-i})}(\st_0) \coloneqq (1 - \discount) \EXP_{(\pol^i, \pol^{-i})}\Biggl[\sum_{t=0}^{\infty}\discount^t\rew^i(\St_t, \Act_t) \Biggm| \St_0 = \st_0\Biggr],
\end{equation}
where the expectation is with respect to the joint measure on all the system variables induced by the choice of the strategy profile of all players.

Although the above model is formulated for an infinite horizon, it can capture interactions for a finite horizon by considering time as part of the state space and by assuming that, at the end of the horizon, the game moves to an absorbing state with zero rewards for all players. In the special case when the game has a single state, a dynamic game is equivalent to an infinitely repeated matrix game. In the special case when the game has only one player, a dynamic game is equivalent to a Markov decision process. 

\begin{remark}\label{rem:normalize}
    We follow the standard game theoretic convention of normalizing the expected total reward by pre-multiplying by $(1-\discount)$. An immediate implication of this is that for any strategy~$\pol$, $\| V^i_{\pol} \|_{\infty} \le \| \rew \|_{\infty}$. In some of the AI literature, the expected reward is not normalized. In such cases
    $\| V^i_{\pol} \|_{\infty} \le \| \rew \|_{\infty}/(1-\discount)$.
\end{remark}

There are two solution concepts commonly used for Markov games, which we state below.
\begin{definition}[Markov perfect equilibrium]\label{def:MPE}
  A Markov strategy profile $\pol = \polAll$,  where $\pol^i \in \Pol^i$, is called a Markov perfect equilibrium (MPE) if for every initial state $\st \in \StSp$, and every player $i \in \playerset$,
  \begin{equation}\label{eq:BR}
    V^i_{(\pol^i, \pol^{-i})}(\st) \ge  V^i_{(\tilde{\pol}^i, \pol^{-i})}(\st), \quad \forall \tilde{\pol}^i \in \Pol^i.
  \end{equation}
\end{definition}
A Markov perfect equilibrium can be viewed as a refinement of subgame perfect
equilibrium where all players play Markov
strategies~\citep{maskin1988theory,maskin1988theory2}. For games with finite
state and action spaces, a Markov perfect equilibrium always
exists~\citep{fink1964equilibrium, RogersNonzerosum1969, VriezeStochastic1987, FilarCompetitive1996}.
For general state and action spaces, see~\cite{TakahashiEquilibrium1964}.

We can also describe MPE equilibrium in terms of best responses. Given a player~$i \in \playerset$ and a strategy profile $\pol^{-i}$ of players other than~$i$, a strategy $\pol^i$ for player is called a best response of $\pol^{-i}$ if it satisfies~\eqref{eq:BR}. We denote this relationship by $\pi^i = \BR^i(\pol^{-i})$. Then, Definition~\ref{def:MPE} is equivalent to stating that a strategy profile $\pol$ is an MPE if, for each player~$i \in \playerset$, $\pol^i$ is a best response of $\pol^{-i}$.

\begin{definition}[Approximate Markov perfect equilibrium]
  Given approximation level $\alpha = (\alpha^i)_{i \in \playerset}$, where $\alpha^i$ are positive constants, a strategy profile $\pol = \polAll$,  where $\pol^i \in \Pol^i$, is called an $\alpha$-approximate Markov perfect equilibrium ($\alpha$-MPE) if for every initial state $\st \in \StSp$, and every player $i \in \playerset$,
  \begin{equation}
    V^i_{(\pol^i, \pol^{-i})}(\st) \ge  V^i_{(\tilde{\pol}^i, \pol^{-i})}(\st) - \alpha^i, \quad \forall \tilde{\pol}^i \in \Pol^i.
  \end{equation}
\end{definition}

When all $\alpha^i$ are identical and equal to say $\alpha'$, we simply call the approximate MPE an $\alpha'$-MPE rather than $(\alpha', \dots, \alpha')$-MPE.

\subsection{Preliminaries on integral probability metrics}
Our results rely on a class of metrics on probability spaces known as  integral probability metrics (IPMs)~\citep{muller1997integral}. 

\begin{definition}\label{def:ipm}
  Let $(\ALPHABET X, \game)$ be a measurable space and $\F$
  denote a class of uniformly bounded measurable functions on $(\ALPHABET X,
  \game)$. The integral probability metric (IPM) between two probability distributions $\mu, \nu \in \PSP(\ALPHABET X)$ with respect to the function class $\F$ is defined as
  \[
    d_{\F}(\mu, \nu) \coloneqq \sup_{f \in \F}\biggl |
    \int_{\ALPHABET{X}} fd\mu - \int_{\ALPHABET{X}} fd\nu \biggr|.
  \]
\end{definition}

Two specific forms of IPMs are used in this paper:
\begin{enumerate}
  \item \textbf{Total variation distance}: If $\F$ is chosen as $\FTV \coloneqq \{f : \Span(f) \le 1\}$, then $d_{\F}$ is the total
    variation distance\footnote{\label{fnt:TV}%
      If $\mu$ and $\nu$ are
      absolutely continuous with respect to some measure $\lambda$ and let $p
      = d\mu/d\lambda$ and $q = d\nu/d\lambda$, then total variation is
      typically defined as $\tfrac 12 \int_{\ALPHABET X}| p(x) - q(x)|
      \lambda(dx)$. This is consistent with our definition. Let $\bar f = ( \sup f + \inf f)/2$. Then
      \begin{align*}
        \hskip 2em & \hskip -2em
        \left| \int_{\ALPHABET X} f d\mu - \int_{\ALPHABET X} f d\nu \right|
        =
        \left| \int_{\ALPHABET X} f(x) p(x) \lambda(dx) 
             - \int_{\ALPHABET X} f(x) q(x) \lambda(dx)
        \right|
        \\
        &= \left| 
        \int_{\ALPHABET X} \bigl[ f(x) - \bar f \bigr] \bigl[ p(x) - q(x) \bigr] \lambda(dx)
        \right|
        \le 
        \| f - \bar f \|_{\infty} \int_{\ALPHABET X} 
        \bigl| p(x) - q(x) \bigr| \lambda(dx)
        \\
        &\le
        \tfrac12 \Span(f) \int_{\ALPHABET X} 
        \bigl| p(x) - q(x) \bigr| \lambda(dx).
    \end{align*}}.
  \item \textbf{Wasserstein distance}: If $\ALPHABET X$ is a metric space and $\F$ is chosen as $\FW \coloneqq \{ f : \Lipschitz(f) \le 1 \}$ (where the Lipschitz constant is computed with respect to the metric on $\ALPHABET X$), then $d_{\F}$ is the Wasserstein distance.
\end{enumerate}
See \cite{subramanian2020approximate} for a discussion of other IPMs and their role in approximate planning for single agent problems. Our approximation results are stated in terms of the Minkowski functional of a function $f$ (not necessarily in $\F$) with respect to a function class $\F$, which is defined as follows:
\begin{equation}
  \Minkowski_{\F}(f) \coloneqq \inf\{
  \Minkowski \in \reals_{> 0} : \Minkowski^{-1} f \in \F \}.
\end{equation}
A key implication of this definition is that for any function $f$,
\begin{equation} \label{eq:minkowski}
  \biggl| \int_{\ALPHABET X} f d\mu - \int_{\ALPHABET X} f d\nu \biggr|
  \le \Minkowski_{\F}(f) \cdot d_{\F}(\mu, \nu).
\end{equation}

The Minkowski functional of the two IPMs considered in this paper are as follows:
\begin{enumerate}
  \item \textbf{Total variation distance}: If $\F$ is chosen as $\FTV$, \( \bigl| \int_{\ALPHABET X} f d\mu - \int_{\ALPHABET X} f d\nu \bigr| \le \Span(f) d_\F(\mu, \nu)\). Thus, for total variation, $\Minkowski_{\FTV}(f) = \Span(f)$.
  \item \textbf{Wasserstein distance}: If $\F$ is chosen as $\FW$, \( \bigl|
    \int_{\ALPHABET X} f d\mu - \int_{\ALPHABET X} f d\nu \bigr| \le \Lipschitz(f)\cdot
    d_\F(\mu,\nu)\). Thus, for the Wasserstein distance, $\Minkowski_{\FW}(f) = \Lipschitz(f)$.
\end{enumerate}

\subsection{Robustness of Markov games to model approximation}
\begin{definition}\label{def:approx-game}
  Given a function class $\F$ and positive constants $\aiserror$, a game $\approxgame \coloneqq \langle \playerset, \StSp, \ActSpAll, \approxTrans, \approxrewAll, \discount \rangle$ is an $\aiserror$-approximation of the game $\game \coloneqq \langle \playerset, \StSp, \ActSpAll, \Trans, \rewAll, \discount \rangle$ if the following conditions are satisfied:
  \begin{enumerate}
    \item \textbf{Reward approximation}: For all $i \in \playerset$, $\st \in \StSp$ and $\act \in \ActSp$,
      \begin{equation}\label{eq:game-rew-approx}
        \lvert \rew^i(\st, \act) - \hat{\rew}^i(\st, \act) \rvert \leq \varepsilon.
      \end{equation}
    \item \textbf{Transition approximation}: For all $\st \in \StSp$ and $\act \in \ActSp$,
     \begin{equation}\label{eq:game-Trans-approx}
       d_{\F}(\Trans(\cdot \mid \st, \act), \approxTrans(\cdot \mid \st, \act)) \le \delta.
     \end{equation}
  \end{enumerate}
\end{definition}

Our main result is the following.

\begin{theorem}\label{thm:main}
  If game $\approxgame$ is an $\aiserror$-approximation of game $\game$ and $\hat{\pol}$ is an $\alpha_{\mathrm{opt}}$-MPE of $\approxgame$, where $\alpha_{\mathrm{opt}} \coloneqq (\alpha_{\mathrm{opt}}^i)_{i \in \playerset}$, then $\hat{\pol}$ is also an $\alpha$-MPE of $\game$, where $\alpha \coloneqq (\alpha^i)_{i \in \playerset}$ can be bounded as
\begin{equation}\label{eq:alpha}
    \alpha^i \le 2 \varepsilon + 
      \frac{\discount}{1-\discount}
      \bigl[ 
        \Delta^i_{(\approxpol^i, \approxpol^{-i})} + \Delta^i_{(\tilde \pol^i_*, \approxpol^{-i})}
      \bigr]
      + \alpha_{\mathrm{opt}}^i,
      \quad \forall i \in \playerset,
\end{equation}
where $\tilde \pol^i_* = \BR^i(\approxpol^{-i})$ and
\[
  \Delta^i_{(\approxpol^i, \approxpol^{-i})} \coloneqq
  \max_{\st \in \St, \act \in \Act} \biggl| \sum_{\st' \in \St} 
    \biggl[ \Trans(\st' | \st, \act) \approxValf^{i}_{(\approxpol^i, \approxpol^{-i})}(\st')
     - \approxTrans(\st' | \st, \act) \approxValf^{i}_{(\approxpol^i, \approxpol^{-i})}(\st') \biggr]
    \biggr|
\]
and $\Delta^i_{(\tilde \pol^i_*, \approxpol^{-i})}$ is defined similarly. Furthermore, we have
\(
  \Delta^i_{(\approxpol^i, \approxpol^{-i})} 
    \le \delta \rho_\F(\approxValf^i_{(\approxpol^i, \approxpol^{-i})}).
\)
and
\(
  \Delta^i_{(\tilde \pol^i_*, \approxpol^{-i})} 
    \le \delta \rho_\F(\approxValf^i_{(\tilde \pol^i_*, \approxpol^{-i})}).
\)
Therefore, a looser upper bound on $\alpha^i$ is given by
\begin{equation}\label{eq:alpha'}
    \alpha^i \le 2 \varepsilon + 
      \frac{\delta \discount}{1-\discount} 
      \bigl[ 
      \rho_\F(\approxValf^i_{(\approxpol^i, \approxpol^{-i})})
      +
      \rho_\F(\approxValf^i_{(\tilde \pol^i_*, \approxpol^{-i})})
      \bigr]
        + \alpha_{\mathrm{opt}}^i,
      \quad \forall i \in \playerset.
  \end{equation}
\end{theorem}
See Sec.~\ref{sec:proof} for the proof.

\begin{remark}
  Both upper bounds~\eqref{eq:alpha} and~\eqref{eq:alpha'} are instance dependent. The bound~\eqref{eq:alpha} depends on the transition probability $\Trans$ of the original game, while the only information about the original game $\game$ needed in~\eqref{eq:alpha'} are the modeling errors $(\varepsilon, \delta)$. 
The approximation bound~\eqref{eq:alpha'} also depends on the choice of IPM
used to measure the approximation in the dynamics. Instance independent bounds
on the approximation error are presented below.
\end{remark}

\subsubsection{Instance independent bounds}

It is possible to obtain instance-independent bounds by using worst case upper bounds on $\rho_\F(\approxValf^i_{\approxpol})$. In particular, the Minkowski functional $\rho_\F(\approxValf^i_{\approxpol})$ is $\Span(\approxValf^i_{\approxpol})$ when using the total variation distance and is $\Lipschitz(\approxValf^i_{\approxpol})$ when using the Wasserstein distance. Using worst-case upper bounds on $\Span(\approxValf^i_{\approxpol})$ and $\Lipschitz(\approxValf^i_{\approxpol})$ gives us the following bounds.

\begin{corollary} \label{cor:TV}
  When $\F = \FTV$, then
  \begin{equation}
    \alpha^i \le 2 \Biggl( \varepsilon + 
      \frac{\discount \delta \Span(\approxrew^i)}{(1-\discount)}
    \Biggr) 
        + \alpha_{\mathrm{opt}}^i,
    \quad \forall i \in \playerset.
  \end{equation}
\end{corollary}

The next bound holds for games where the transition matrix and reward function are Lipschitz. 
\begin{definition}
 Suppose the state space $\StSp$ is a metric space with metric~$d$. Then, a game $\game$ is said to be $(L_r, L_{\Trans})$-Lipschitz if for any $i \in \playerset$, $\st_1, \st_2 \in \StSp$ and $\act \in \ActSp$, 
\begin{align*}
  \bigl| \rew^i(\st_1, \act) - \rew^i(\st_2, \act) \bigr| 
  &\le L_r  d(\st_1, \st_2),\\
  \shortintertext{and}
  \W(\Trans(\cdot | \st_1, \act), \Trans(\cdot | \st_2, \act))
  &\le L_{\Trans} d(\st_1, \st_2),
\end{align*}
where $\W$ denotes the Wasserstein distance. 
\end{definition}

Furthermore, define
\[
    \Lipschitz(\approxpol^i) = 
    \sup_{\substack{\st, \st' \in \StSp \\ \st \neq \st' }}
    \frac{ d_{\FW}(\approxpol^i(\cdot | \st), \approxpol^i(\cdot | \st')) }
    { d(\st, \st') }.
\]

\begin{corollary}\label{cor:W}
   When $\F = \FW$ and $\approxgame$ is $(L_\rew, L_\Trans)$-Lipschitz with $\discount L_\Trans < 1$ and $\discount (1 + \Lipschitz(\approxpol^i)) L_\Trans < 1$ for all $i \in \playerset$, then
    \begin{equation}
        \alpha^i \le 2 \varepsilon + 
        \frac{\discount L_r \delta}{(1-\discount L_{\Trans})}
        + 
        \frac{\discount L_r \delta}{(1-\discount(1 + \Lipschitz(\approxpol^i)) L_{\Trans})}
        + \alpha_{\mathrm{opt}}^i,
      \quad \forall i \in \playerset.
    \end{equation}
    Furthermore, when $\alpha_{\mathrm{opt}} = 0$, we have
    \begin{equation}
        \alpha^i \le 2 \Biggl( \varepsilon + 
        \frac{\discount L_r \delta}{(1-\discount L_{\Trans})}
      \Biggr) 
      \quad \forall i \in \playerset.
    \end{equation}
\end{corollary}

\begin{remark}
    The result of Theorem~\ref{thm:main} and Corollaries~\ref{cor:TV} and~\ref{cor:W} imply that MPE is continuous in model approximation for total variation and Wasserstein distances. To show this, we first start with the case when $\mathfrak{F} = \FTV$. Consider a sequence $\{\approxgame_n\}_{n \ge 0}$ of approximations such that $\approxgame_n$ is an $(\varepsilon_n, \delta_n)$ approximation of $\game$, where $\varepsilon_n \to 0$ and $\delta_n \to 0$ as $n \to \infty$. Let $\approxpol_n$ be an MPE of $\approxgame_n$. Then, Theorem~\ref{thm:main} shows that $\approxpol_n$-is an $\alpha_n$-MPE of $\game$, where $\alpha_n$ is upper bounded by $2\varepsilon_n + 2 \discount \delta_n \Span(\approxrew^i_n)/(1-\discount)$. Since $\Span(\approxrew^i_n)$ is bounded, $\alpha_n \to 0$ as $n \to \infty$. Thus, any limit point $\approxpol_{\infty}$ of $\{\approxpol_n\}_{n \ge 0}$ is an MPE of game~$\game$. The same argument applies for the case when $\mathfrak{F} = \FW$, provided that each approximate game~$\approxgame_n$ is $(L_r, L_p)$ Lipschitz with $\discount L_p < 1$. 
\end{remark}

\begin{remark}
    Although we have only elaborated on two specific choices of IPMs (total variation and Wasserstein distances), the result of Theorem~\ref{thm:main} is applicable for \emph{any} IPM. Many other IPMs have been considered in the literature including Kolmogorov distance, bounded Lipschitz metric, and maximum mean discrepancy. See, for example, \cite{muller1997integral, subramanian2020approximate}. The choice of the metric often depends on the specific properties of the model.
\end{remark}

\subsection{Model based RL for Markov games}\label{sec:MARL}
Now, we consider the setting where the components $\langle \StSp,
(\ActSp^i)_{i \in \playerset},
(\rew^i)_{i \in \playerset}, \discount \rangle$ of a game $\game$
are known but the transition
probability matrix $\Trans$ is not known. Suppose we have access to a
\emph{generative model}, i.e., a black-box simulator which provides samples
$\St_{+} \sim \Trans(\cdot | \st, \act)$ of the next state $\St_{+}$ for a
given state-action pair $(\st, \act)$ as input. Suppose we call the simulator
$n$ times at each state-action pair and estimate an empirical model
$\approxTrans_n$ as $\approxTrans_n(s' | s, a) \coloneqq \text{count}(s'|s,a)/n$, where
$\text{count}(s'| s,a)$ is the number of times $s'$ is sampled with the input
is $(s,a)$. The game $\approxgame_n \coloneqq \langle \StSp, (\ActSp^i)_{i \in
\playerset}, \approxTrans_n, (\rew^i)_{i \in \playerset},
\discount \rangle$ may be viewed as an approximation of game $\game$. 
We further assume that there is a planning oracle, which takes the approximate game $\approxgame_n$ as input and generates an $\alpha_{\mathrm{opt}}$-MPE $\approxpol_n$, where $\alpha_{\mathrm{opt}} \in \reals_{> 0}$ is a property of the planning oracle.

Note that we are assuming that there is a system planner which generates samples from the generative model and computes an  $\alpha_{\mathrm{opt}}$-MPE of $\approxgame_n$. A more interesting setting is where each player generates independent samples from the generative model and computes a different MPE. This setting is challenging due to the multiplicity of MPE and is not discussed in this paper.

One fundamental question in this setting is the following. Given an $\alpha >
0$ and $p > 0$, how many samples~$n$ are needed from the generative model to ensure that
$\approxpol_n$ is an $\alpha$-MPE for game $\game$ with probability at least $1 - p$. This is called the \emph{sample
complexity} of learning. Below, we obtain two bounds on sample
complexity using the approximation bounds of Theorem~\ref{thm:main}.

\begin{theorem}[Hoeffding-Type Bound]\label{thm:MARL}
  For any $\alpha_{\mathrm{opt}} > 0$, $\alpha > \alpha_{\mathrm{opt}}$ and $p > 0$, let
  \begin{align} \label{eq:nH:MARL}
    n \ge n_H(\discount) \coloneqq \biggl\lceil \biggl( \frac{ \discount }{(1 - \discount)^2} \Span(r) \biggr)^2 
    \frac{ 2 \log(2 |\StSp|\bigl(\prod_{i \in \playerset}|\ActSp^i| \bigr) |\playerset| p^{-1}) }
    {(\alpha - \frac{1+\discount}{1-\discount}\alpha_{\mathrm{opt}})^2} \biggr\rceil.
  \end{align}
  Then, with probability $1 - p$, any $\alpha_{\mathrm{opt}}$-MPE $\approxpol_n$ of game
  $\approxgame_n$ is an $\alpha$-MPE for game $\game$.
\end{theorem}
\begin{theorem}[Bernstein-Type Bound]\label{thm:MARL2}
  For any $\alpha_{\mathrm{opt}} > 0$, $\alpha \in (5\alpha_{\mathrm{opt}}$, $\|r\|_{\infty} \sqrt{1-\gamma} + 5 \alpha_{\mathrm{opt}})$ and $p > 0$, let
  \[
    n \ge n_B(\discount) \coloneqq \biggl\lceil 
    \frac{c\discount \|r\|_{\infty}^2\log(8|\StSp|\bigl(\prod_{i \in \playerset}|\ActSp^i| \bigr)|\playerset|(1-\discount)^{-2}p^{-1})}{(1-\discount)(\alpha - 5\alpha_{\mathrm{opt}})^2}
    \biggr\rceil,
  \]
  where $c$ is an absolute constant
  Then, with probability $1 - p$, any $\alpha_{\mathrm{opt}}$-MPE $\approxpol_n$ of game
  $\approxgame_n$ is an $\alpha$-MPE for game $\game$. 
\end{theorem}
See Sec.~\ref{sec:proof} for the proofs.
\begin{remark}
  In general, game $\approxgame_n$ may have multiple MPEs. The results of Theorems~\ref{thm:MARL} and~\ref{thm:MARL2} are true for \emph{every} MPE of game $\approxgame_n$.
\end{remark}

\begin{remark}
    We have assumed that we generate $n$ samples for for every state-action pair. Therefore, the total number of samples required in Theorems~\ref{thm:MARL} and~\ref{thm:MARL2} are $n |\StSp|\, \prod_{i=1}^n|\ActSp^i|$. 
\end{remark}

\subsection{Discussion of the results}
\subsubsection{Bounds on the values of the Markov perfect equilibrium}

The robustness and sample complexity bounds presented above only show that an $\alpha_{\mathrm{opt}}$-MPE $\approxpol$ of approximate game~$\approxgame$ is an $\alpha$-MPE of the original game~$\game$. A natural question is whether the value of policy $\approxpol$ in game~$\game$ is close to the value of \emph{some} MPE of $\game$. We do not know an answer to this question in general, but we provide an answer for the special case of two-player zero sum games (ZSG). 

\begin{definition}
A two player Markov game is called \emph{zero-sum} if
\[
   r^1(s,a) + r^2(s,a) = 0, \quad \forall s \in \StSp, a \in \ActSp.
\]
\end{definition}
Two important properties of zero-sum games are the following \citep{FilarCompetitive1996}:
\begin{enumerate}
    \item[\textbf{(P1)}] For any Markov policy $\pol$, 
    \[
      V^1_\pol(s) + V^2_\pol(s) = 0, \quad \forall s \in \StSp.
    \]
    \item[\textbf{(P2)}] All MPE of ZSG have the same value, i.e., if $\pol_*$ and $\tilde \pol_*$ are MPEs of a ZSG, then
    \[
       V^i_{\pol_*}(s) = V^i_{\tilde \pol_*}(s), \quad i \in \{1,2\}, s \in \StSp.
    \]
    We use $V^i_*$ to denote the \emph{value of the game}. 
\end{enumerate}

We now establish an important property of approximate MPEs in ZSGs.

\begin{proposition}\label{prop:ZSG}
  If $\pol$ is an $\alpha$-MPE of a two-player ZSG with value $V_* = (V^1_*,V^2_*)$, where $\alpha = (\alpha^1, \alpha^2)$, then 
  \begin{equation}
    \label{eq:ZSG}      
     V^i_\pol(s) \ge V^i_*(s) - \alpha^i, \quad \forall s \in \StSp.
  \end{equation}
\end{proposition}
\begin{proof}
  Let $\pol_* = (\pol^1_*, \pol^2_*)$ be any MPE of $\game$.
  For any $s \in \StSp$, we have the following:
  \begin{align*}
     V^1_\pol(s) &\stackrel{(a)}\ge  V^1_{(\pol^1, \pol^2_*)}(s) - \alpha^1
     \\
     &\stackrel{(b)}= - V^2_{(\pol^1, \pol^2_*)}(s) - \alpha^1
     \\
     &\stackrel{(c)}\ge - V^2_{(\pol^1_*, \pol^2_*)}(s) - \alpha^1
     \\
     &\stackrel{(d)}=  V^1_{(\pol^1_*, \pol^2_*)}(s) - \alpha^1
     \\
     &\stackrel{(e)}=  V^1_*(s) - \alpha^1
  \end{align*}
  where $(a)$ follows from the definition of $\alpha$-MPE, $(b)$ and $(d)$ follow from (P1), and $(c)$ follows from  the fact that $\pol_*$ is an MPE, and $(e)$ follows from (P2). This proves~\eqref{eq:ZSG} for $i = 1$. The result for $i=2$ follows from a similar argument. 
\end{proof}

Combining Proposition~\ref{prop:ZSG} with the robustness results of Theorem~\ref{thm:main}, we get the following.
\begin{corollary}
  \label{cor:main}
  If the game~$\game$ in Theorem~\ref{thm:main} is a two-player ZSG with value $V_* = (V^1_*, V^2_*)$, then under the conditions of Theorem~1, we have
  \[
    V^i_{\approxpol}(s) \ge V^i_*(s) - \alpha^i, \quad i \in \{1,2\}, s \in \StSp
  \]
  where $\alpha^i$ is given by~\eqref{eq:alpha}.
\end{corollary}

Similarly, combining Proposition~\ref{prop:ZSG} with Theorems~\ref{thm:MARL} and~\ref{thm:MARL2}, we get the following.
\begin{corollary}\label{cor:MARL}
  If the game~$\game$ is a two-player ZSG with value $V_* = (V^1_*, V^2_*)$, then under the conditions of Theorem~\ref{thm:MARL} or Theorem~\ref{thm:MARL2}, we have
  \[
    V^i_{\approxpol_n}(s) \ge V^i_*(s) - \alpha, \quad i \in \{1,2\}, s \in \StSp.
  \]
\end{corollary}

Note that Corollary~\ref{cor:MARL} under the conditions of Theorem~\ref{thm:MARL2} is identical to Thereom~3.2 of \cite{zhang2020model}. Thus, the results of our paper provide an alternative proof of the sample complexity results for zero-sum games. 

\subsubsection{Discussion on the sample complexity bounds}\label{sec:comparison}

The bounds of Theorem~\ref{thm:MARL} and~\ref{thm:MARL2} provide an upper bound on sample complexity for learning MPE. 
As argued in~\cite{zhang2020model}, the lower bound on the sample complexity for MDPs obtained in \cite{azar2013minimax} (see Remark~\ref{rem:tight} later) can be directly translated to games to provide a lower bound of 
\[
  \Omega\biggl(|\StSp|\Bigl(\prod_{i=1}^n|\ActSp^i|\Bigr)\frac{\log(|\StSp|\bigl(\sum_{i\in \playerset}|\ActSp^i|\bigr)p^{-1})}{(1-\discount)\alpha^2} \biggr).
\]
Thus, upper bounds of Theorems~\ref{thm:MARL} and~\ref{thm:MARL2} are tight in $|\StSp|$, but loose in the logarithmic factor in $\{|\ActSp^i|\}$ ($\sum_{i \in \playerset} |\ActSp^i|$ vs $\prod_{i \in \playerset} |\ActSp^i|$). 
The sample complexity bounds for zero-sum games obtained in \cite{zhang2020model} were also loose in the action space. \cite{zhang2020model} also highlight the difficulty in tightening the lower bound to $\Omega(\log(\prod_{i\in\playerset}|\ActSp^i|))$. Since, zero-sum games are a special case of general-sum games, the same difficulties hold for the general-sum games as well. 

In addition, the bound of Theorem~\ref{thm:MARL2} is tight in the discount factor while the bound of Theorem~\ref{thm:MARL} is loose by a factor of $1/(1-\discount)^3$. This means that as the discount factor $\discount \to 1$, the bound of Theorem~\ref{thm:MARL2} is tighter. Therefore, we must have the following.
\begin{proposition} \label{prop:sample_complexity_comp}
   There exists a critical discount factor $\discount^\circ$ (perhaps depending on $|\StSp|$ and $\{|\ActSp^i|\}_{i \in \playerset}$) such that for all $\discount \ge \discount^\circ$, $n_H(\discount) \ge n_B(\discount)$, where $n_H(\discount)$ and $n_B(\discount)$ are as defined in Theorems~\ref{thm:MARL} and~\ref{thm:MARL2}, respectively.
\end{proposition}

It is not possible to derive a closed form expression on $\discount^\circ$, but we can obtain a lower bound as follows.

\begin{lemma}\label{lem:gamma-bar}
 Let $\alpha_{\mathrm{opt}} = 0$ and let $\bar{\discount}$ be such that $c = 8\bar{\discount}/(1-\bar{\discount})^3$. Then for all $\discount \le \bar{\discount}$, we have $n_H(\discount) \le n_B(\discount)$. 
\end{lemma}
\begin{proof}
    Consider the following sequence of inequalities (for $\alpha_{\mathrm{opt}} = 0$)
    \begin{align*}
     n_H(\gamma) 
    & \stackrel{(a)}{\le} \frac{ 8\discount^2 \|r\|_\infty^2  \log(2 |\StSp|\bigl(\prod_{i \in \playerset}|\ActSp^i| \bigr) |\playerset| p^{-1}) }{(1-\discount)^4} \alpha^2 
    \\
    &\stackrel{(b)} \leq \frac{8\discount}{(1-\discount)^3} \cdot \frac{\discount \|r\|_\infty^2  \log(8 |\StSp|\bigl(\prod_{i \in \playerset}|\ActSp^i| \bigr) |\playerset| (1-\discount)^{-2} p^{-1}) }{(1-\discount) \alpha^2}
    \\
    &\stackrel{(c)} \leq c \cdot \frac{\discount \|r\|_\infty^2  \log(8 |\StSp|\bigl(\prod_{i \in \playerset}|\ActSp^i| \bigr) |\playerset| (1-\discount)^{-2} p^{-1}) }{(1-\discount) \alpha^2} \\
    &= n_B(\discount),
    \end{align*}
    where $(a)$ uses the fact that $\Span(r) \le 2 \|r\|_\infty$, $(b)$ follows from the fact that the additional multiplier in the $\log$ term is greater than $1$ and $(c)$ follows from the fact that $8\discount/(1-\discount)^3$ is increasing in $\discount$ for $\discount \in (0,1)$.
   \qed
\end{proof}

We now present a lower bound on $c$, which implies a lower bound on $\discount^\circ$.
\begin{lemma}\label{lemma:c}
  A lower bound on the constant $c$ in Theorem~\ref{thm:MARL2} is given by
  \(
     c \ge 64
  \). 
  Therefore, when $\alpha_{\mathrm{opt}} = 0$, we get a lower bound $\bar{\discount}=\GAMMAN$ on $\discount^\circ$ by solving $64 = 8\discount/(1-\discount)^3$. And thus, for all $\discount \le \GAMMAN$, $n_H(\discount) \le n_B(\discount)$.
\end{lemma}
\begin{proof}
The exact calculation of $c$ is slightly nuanced, but a simple upper bound can be obtained as follows. The bound in Theorem~\ref{thm:MARL2} relies on \cite[Theorem 1]{agarwal2020model} (stated as Theorem~\ref{thm:generative-kakade} later in this paper). In the proof of \cite[Theorem 1]{agarwal2020model}, the number of samples $n$ are chosen such that $(\gamma/(1 - \gamma))\sqrt{8L'/n} \le 1/2$, where $L' = \log(8|\StSp| |\ActSp| (1-\discount)^{-2} p^{-1})$ (also see the proof of Theorem~\ref{thm:generative-kakade}).  This implies that $n \ge 64 \gamma^2 L'/(1-\gamma)^2$. Comparing this with the bound on $n$ in Theorem~\ref{thm:MARL2} (and, more precisely, the bound in Theorem~\ref{thm:generative-kakade} later), we get that $c \ge 64$. The bound on $\discount^\circ$ now follows from Lemma~\ref{lem:gamma-bar}.
\qed
\end{proof}

\begin{remark}
  The result of Lemma~\ref{lemma:c} shows that for all $\discount \le \GAMMAN$, the sample complexity bound of Theorem~\ref{thm:MARL} is tighter than the sample complexity bound of Theorem~\ref{thm:MARL2}. Note that the actual value of the critical discount factor $\discount^\circ$ is higher, and depends on the values of $|\StSp|$, $|\ActSp^i|$ and $|\playerset|$. See Sec.~\ref{sec:numerical_examples} for an example.
\end{remark}

\section{Background on MDPs}\label{sec:MDP}
The main idea for proving the results of Sec.~\ref{sec:model} is to pick a player, say $i$, fix the strategy profile of all players other than $i$, and then look at the best response of player~$i$. The problem of finding the best response at player~$i$ is a single agent Markov decision process. So, we start by reviewing the pertinent results from Markov decision theory. All the results in this section are either standard or variations of existing results.
\subsection{MDP, Bellman Operators, and Dynamic Programming}
A Markov Decision Process (MDP) is a tuple $\langle \StSp, \ActSp, \Trans, \rew, \discount \rangle$ where 
\begin{itemize}
  \item $\StSp$ is the (finite) set of states of the environment. The state at time~$t$ is denoted by $\St_t$.
  \item $\ActSp$ is the (finite) set of actions available to the agent. The action at time~$t$ is denoted by $\Act_t$.
    \item $\Trans : \StSp \times \ActSp \to \PSP(\StSp)$ is the controlled transition probability. For any realization $\st_{1:t+1}$ of $\St_{1:t+1}$ and choice $\act_{1:t}$ of $\Act_{1:t}$, we have
    \begin{equation}
      \Trans( \st_{t+1}| \st_t, \act_t) \coloneqq \PR(\St_{t+1} = \st_{t+1} |\St_{1:t} = \st_{1:t}, \Act_{1:t}=\act_{1:t}).
    \end{equation}
    \item $\rew: \StSp \times \ActSp \to \reals$ is the per-step reward function.
    \item $\discount \in (0, 1)$ is the discount factor. 
\end{itemize}
It is assumed that the agent observes the state $\St_t$ and chooses the action
$\Act_t$ according to a Markov strategy $\pol : \StSp \to \PSP(\ActSp)$.
The performance of a Markov  strategy $\pol$ starting from initial state
$\st_0 \in \StSp$ is given by:\footnote{For consistency with the normalized rewards considered in the game formulation (see Remark~\ref{rem:normalize}), we use normalized rewards for MDPs as well. Although most of the literature on MDPs uses unnormalized rewards, normalized rewards are commonly used in the literature on constrained MDPs \citep{altman1999constrained}.}
\begin{equation}
  \Valf_{\pol}(\st_0) \coloneqq (1 - \discount) \EXP_\pol\biggl[\sum_{t=0}^{\infty} \discount^t \rew(\St_t, \Act_t)\biggm|\St_0 = \st_0\biggr],
\end{equation}
where the expectation is with respect to the joint measure on the system variables induced by the choice of strategy~$\pol$. 
A strategy~$\pol$ is called optimal if for any other Markov strategy~$\tilde{\pol}$, we have
\begin{equation}
    \Valf_\pol(\st) \ge \Valf_{\tilde{\pol}}(\st), \quad \forall \st \in \StSp.
\end{equation}
In addition, given a positive constant $\alpha$, a strategy $\pol$ is called $\alpha$-optimal if 
\begin{equation}
    \Valf_\pol(\st) \ge \Valf_{\tilde{\pol}}(\st) - \alpha, \quad \forall \st \in \StSp.
\end{equation}

Given an MDP $\MDP \coloneqq \langle \StSp, \ActSp, \Trans, \rew, \discount \rangle$ and a Markov strategy~$\pol$, define the Bellman operators $\Bellman_\pol: \reals^{|\StSp|} \to \reals^{|\StSp|}$ and $\Bellman_*: \reals^{|\StSp|} \to \reals^{|\StSp|}$ as follows: for any $\valf \in \reals^{|\StSp|}$ and $\st \in \StSp$
\begin{align}
        [\Bellman_\pol\valf](\st) 
        &\coloneqq \sum_{\act \in \ActSp}\pol(\act|\st)\Bigl[(1 - \discount) \rew(\st, \act) + \discount\sum_{\st' \in \StSp}\Trans(\st'|\st, \act)\valf(\st')\Bigr], 
        \\
        [\Bellman_* \valf](\st) 
        &\coloneqq \max_{\act \in \ActSp}\Bigl[(1 - \discount) \rew(\st, \act) 
        + \discount\sum_{\st' \in \StSp}\Trans(\st'|\st, \act)\valf(\st')\Bigr]. 
\end{align}
Then, optimal and approximately optimal strategies can be characterized using the Bellman operators as shown below. These are standard results. See~\cite{bertsekas_book}, for example.

\begin{proposition}\label{prop:optimal-strategy-MDP}
A Markov strategy $\pol$ is optimal if and only if there exists a value function $\Valf \in \reals^{|\StSp|}$ such that
\begin{equation} \label{eq:V-MDP}
    \Valf = \Bellman_\pol\Valf 
    \quad\hbox{and}\quad
    \Valf = \Bellman_*\Valf.
\end{equation}
\end{proposition}
\begin{remark}
  Note that an MDP can have more than one optimal strategy but all optimal strategies have the same performance and hence the same value function. 
\end{remark}
An immediate implication of Proposition~\ref{prop:optimal-strategy-MDP} and the definition of an $\alpha$-optimal strategy is the following.
\begin{proposition}\label{prop:approx-optimal-strategy-MDP}
  Given a Markov strategy $\pol$, let $\Valf_\pol$ be the unique fixed point of 
  \( \Valf_\pol = \Bellman_\pol \Valf_\pol \) and let $\Valf_*$ be the unique fixed point of $\Valf_* = \Bellman_* \Valf_*$. Then, the strategy~$\pol$ is $\alpha$-optimal if and only if
  \begin{equation}\label{eq:alpha-V-MDP}
    \Valf_\pol \ge \Valf_* - \alpha \ONES.
  \end{equation}
\end{proposition}
\begin{proof}
 This is an immediate implication of Proposition~\ref{prop:optimal-strategy-MDP} and the definition of $\alpha$-optimal strategy. We provide a proof for completeness.
    \begin{enumerate}
        \item[($\Rightarrow$)] Let $\pol^*$ be an optimal policy and $\pol$ be an $\alpha$-optimal strategy. Then, by definition of an $\alpha$-optimal strategy, we have $\Valf_{\pol}(\st) \ge \Valf_{\pol^*}(\st) - \alpha$, for all $\st \in \StSp$. Moreover, since $\pol^*$ is optimal,  we have $\Valf_{\pol^*}(\st) = \Valf_{*}(\st)$. Combining the two, we have
        $\Valf_{\pol}(\st) \ge \Valf_{*}(\st) - \alpha$, which is same as~\eqref{eq:alpha-V-MDP}.
        \item[($\Leftarrow$)] Let $\pol$ be a policy which satisfies~\eqref{eq:alpha-V-MDP} and $\tilde{\pol}$ be any policy. By definition of optimal policy, we have $\Valf_{*}(\st) \ge \Valf_{\tilde \pol}(\st)$ for all $\st \in \StSp$. Therefore, from~\eqref{eq:alpha-V-MDP}, we have that $\Valf_{\pol}(\st) \ge \Valf_{\pol^*}(\st) - \alpha \ge \Valf_{\tilde \pol}(\st) - \alpha$. Since this inequality holds for every $\tilde \pol$, the policy $\pol$ is $\alpha$-optimal.
        \qed
    \end{enumerate}
\end{proof}

\begin{remark}\label{rem:suff-con-alphaMDP}
  A sufficient condition to verify~\eqref{eq:alpha-V-MDP} in Proposition~\ref{prop:approx-optimal-strategy-MDP} is that
\begin{equation}
  \Valf_\pol \ge \Bellman_* \Valf_\pol -  \alpha \ONES.
\end{equation}
\end{remark}

We now present some basic properties of the value function which are used later.

\begin{lemma}\label{lem:v-span}
  If $\Valf$ is the optimal value function of MDP $\MDP$, then 
  \[
    \Span(\Valf) \le \Span(r).
  \]
\end{lemma}
\begin{proof}
  This result follows immediately by observing that the per-step reward
  $\rew(\St_t, \Act_t) \in [ \min(\rew), \max(\rew) ]$. Therefore, $\max(V) \le
  \max(\rew)$ and $\min(V) \ge \min(\rew)$.
  \qed
\end{proof}

We now define the notion of a Lipschitz MDP.
\begin{definition}
Let $d$ be a metric on the state space $\StSp$. 
The MDP $\MDP$ is said to be $(L_r, L_{\Trans})$-Lipschitz if for any $\st_1, \st_2 \in \StSp$ and $\act \in \Act$, the reward function $r$ and transition kernel $\Trans$ of $\MDP$ satisfy the following
\begin{align*}
  \bigl| r(\st_1, \act) - r(\st_2, \act) \bigr| 
  &\le L_r d(\st_1, \st_2),\\
  \shortintertext{and}
  \W(\Trans(\cdot | \st_1, \act), \Trans(\cdot | \st_2, \act))
  &\le L_{\Trans} d(\st_1, \st_2),
\end{align*}
where $\W$ denotes the Wasserstein distance. 
\end{definition}

\begin{lemma} \label{lemma:lip_policy}
  If an MDP $\MDP$ is $(L_\rew, L_\Trans)$-Lipschitz, then for any policy $\pol$, the corresponding value function $\Valf_{\pol}$ is Lipschitz with
  \[
    \Lipschitz(\Valf_{\pol}) \le \frac{(1-\discount) L_r}{ 1 - \discount(1 + \Lipschitz(\pi)) L_{\Trans}},
  \]
  provided $\discount (1 + \Lipschitz(\pol)) L_{\Trans} < 1$, where $\Lipschitz(\pol)$ is the Lipschitz-constant of the strategy $\pol$, i.e., 
  \[
     \Lipschitz(\pol) = \sup_{\substack{\st,\st' \in \StSp \\ \st \neq \st'}}
     \frac{d_{\FW}(\pol(\cdot | \st), \pol(\cdot | \st')}{d(\st,\st')}.
  \]
\end{lemma}
\begin{proof}
  The result follows from~\cite[Theorem 4.1]{Hinderer2005}.
  \qed
\end{proof}

The above result can be strengthened when the policy $\pol$ is the optimal policy.

\begin{lemma} \label{lemma:mink_lip}
  If an MDP $\MDP$ is $(L_\rew, L_\Trans)$-Lipschitz and $\discount L_{\Trans} < 1$, and $\Valf_*$ is the optimal value function of $\MDP$, then 
  \[
    \Lipschitz(\Valf_*) \le \frac{(1-\discount) L_r}{ 1 - \discount L_{\Trans}}.
  \]
\end{lemma}
\begin{proof}
  The result follows from~\cite[Theorem 4.2]{Hinderer2005}.
  \qed
\end{proof}

\subsection{Robustness of MDPs to model approximation}
\begin{definition}\label{def:approxMDP}
Given a function class $\F$ and positive constants $\aiserror$, we say that an MDP $\approxMDP \coloneqq \langle \StSp, \ActSp, \approxTrans, \approxrew, \discount \rangle$ is an $\aiserror$-approximation of the MDP $\MDP \coloneqq \langle \StSp, \ActSp, \Trans, \rew, \discount \rangle$ if it satisfies the following properties:
\begin{enumerate}
  \item \textbf{Reward approximation}: For all $\st \in \StSp$, and $\act \in \ActSp$,
    \begin{equation}
        \lvert \rew(\st, \act) - \approxrew(\st, \act) \rvert \le \varepsilon.
    \end{equation}
  \item \textbf{Transition approximation}: For all $\st \in \StSp$, and $\act \in \ActSp$,
    \begin{equation}
        d_{\F}(\Trans(\cdot | \st, \act), \approxTrans(\cdot | \st, \act)) \le \delta.
    \end{equation}
\end{enumerate}
\end{definition}

The main approximation result for MDPs relevant for our analysis is the following.

\begin{theorem}\label{thm:MDP-approx}
  Given a function class $\F$ and an MDP $\MDP \coloneqq \langle \StSp, \ActSp, \Trans, \rew, \discount \rangle$, 
  suppose $\approxMDP \coloneqq \langle \StSp, \ActSp, \approxTrans, \approxrew, \discount \rangle$ is an $\aiserror$-approximation of~$\MDP$. Let $\approxpol_*$ and $\approxpol$ be an an optimal and an $\alpha_{\mathrm{opt}}$-optimal strategy of $\approxMDP$. Let $\approxValf_{\approxpol_*}$ and $\approxValf_{\approxpol}$ be the corresponding value functions.
  Then $\approxpol$ is an $\alpha$-optimal strategy of $\MDP$ with 
\begin{equation}\label{eq:MDP-bound}
  \alpha \le 2 \varepsilon + \frac{\discount}{1-\discount} \bigl[ \Delta_{\approxpol_*} + \Delta_{\approxpol} \bigr] + \alpha_{\mathrm{opt}},
\end{equation}
where
\[
  \Delta_{\approxpol} \coloneqq
  \max_{\st \in \St, \act \in \Act} \biggl| \sum_{\st' \in \St} 
    \biggl[ \Trans(\st' | \st, \act) \approxValf_{\approxpol}(\st')
     - \approxTrans(\st' | \st, \act) \approxValf_{\approxpol}(\st') \biggr]
    \biggr|
\]
and $\Delta_{\approxpol_*}$ is defined similarly. Furthermore, we can show that
\[
   \Delta_{\approxpol} \le \delta \rho_\F(\approxValf_{\approxpol})
\]
and a similar bound holds for $\Delta_{\approxpol_*}$. 
\end{theorem}
\begin{proof}
  The bound for \eqref{eq:MDP-bound} is established in Appendix~\ref{app:MDP}. The upper bound on $\Delta_{\approxpol}$ and $\Delta_{\approxpol_*}$ follows from~\eqref{eq:minkowski} and definition of $\delta$. 
  \qed
\end{proof}

\begin{theorem}\label{cor:MDP-approx}
  Under the setup of Theorem~\ref{thm:MDP-approx}, the approximation gap $\alpha$ can also be bounded by
  \begin{equation}
      \alpha \le \frac{(2 + \discount)}{(1 - \discount)} \varepsilon + \frac{2\discount}{(1 - \discount)^2} \overline {\Delta}_{\pol_*} + \frac{(1+\discount)}{(1-\discount)}\alpha_{\mathrm{opt}}
  \end{equation}
  where $\pol_*$ is the optimal policy for model $\MDP$ and
\[
  \overline{\Delta}_{\pol_*} \coloneqq
  \max_{\st \in \St, \act \in \Act} \biggl| \sum_{\st' \in \St} 
    \biggl[ \Trans(\st' | \st, \act) \Valf_{\pol_*}(\st')
     - \approxTrans(\st' | \st, \act) \Valf_{\pol_*}(\st') \biggr]
    \biggr|.
\]
\end{theorem}
\begin{proof}
See Appendix~\ref{app:MDP2} for details.
\end{proof}

Note that the results of Theorem~\ref{thm:MDP-approx} are instance dependent. We present instance independent bounds on $\alpha$ by upper bounding $\rho_{\F}(V)$.
\subsubsection{Instance independent bounds}

\begin{corollary}\label{cor:ais-tv}
  If the function class $\F$ in Theorem~\ref{thm:MDP-approx} is $\FTV$, then 
  \[
    \alpha \le 2\Bigl(\varepsilon + \frac{\discount \delta \Span(\approxrew) }{(1-\discount)} \Bigr) 
    + \alpha_{\mathrm{opt}}.
\]
\end{corollary}
\begin{proof}
  The result follows from the observation that $\rho_{\TV}(\approxValf) = \Span(\approxValf)$ and then using Lemma~\ref{lem:v-span} in Theorem~\ref{thm:MDP-approx}.
  \qed
\end{proof}

\begin{corollary}\label{cor:ais-w}
  If the function class $\F$ in Theorem~\ref{thm:MDP-approx} is $\FW$, and the approximate MDP $\approxMDP$ is $(L_\rew, L_\Trans)$-Lipschitz with $\discount L_\Trans < 1$ and $\discount (1 + \Lipschitz(\approxpol)) L_{\Trans} < 1$, then  
  \[
    \alpha \le 2 \varepsilon + 
    \frac{\discount \delta L_r }{(1-\discount L_{\Trans})}
    +
    \frac{\discount \delta L_r }{(1 - \discount (1 + \Lipschitz(\approxpol)) L_{\Trans})}
    + \alpha_{\mathrm{opt}}.
\]
  Furthermore, if $\alpha_{\mathrm{opt}} = 0$, the above expression simplifies to
  \[
    \alpha \le 2\Biggl( \varepsilon + 
    \frac{\discount \delta L_r }{(1-\discount L_{\Trans})}
    \Biggr).
  \]
\end{corollary}
\begin{proof}
  The result follows from the observation that $\rho_{\FW}(\approxValf) = \Lipschitz(\approxValf)$ and then using Lemmas~\ref{lemma:lip_policy} and~\ref{lemma:mink_lip} in Theorem~\ref{thm:MDP-approx}.
  \qed
\end{proof}

\subsection{Model based RL for MDPs}

In this section, we consider a setting similar to Sec.~\ref{sec:MARL}, but for
MDPs. Suppose the components $\langle \StSp, \ActSp,
\rew, \discount \rangle$ of an MDP $\MDP$ are known but the transition
probability matrix $\Trans$ is not known. Similar to the game setting, we assume that we have access to a
\emph{generative model}, i.e., a black-box simulator which provides samples
$\St_{+} \sim \Trans(\cdot | \st, \act)$ of the next state $\St_{+}$ for a
given state-action pair $(\st, \act)$ as input. Suppose we call the simulator
$n$ times at each state-action pair and estimate an empirical model
$\approxTrans_n$ as $\approxTrans_n(s' | s, a) \coloneqq \text{count}(s'|s,a)/n$, where
$\text{count}(s'| s,a)$ is the number of times $s'$ is sampled with the input
is $(s,a)$. The MDP $\approxMDP_n \coloneqq \langle \StSp, \ActSp, \approxTrans_n, \rew,
\discount \rangle$ may be viewed as an approximation of MDP $\MDP$. As in the game setting, we assume that there is a planning oracle, which takes the approximate model $\approxMDP_n$ as input and generates an $\alpha_{\mathrm{opt}}$-optimal strategy $\approxpol_n$, where $\alpha_{\mathrm{opt}} \in \reals_{> 0}$ is a property of the planning oracle.

As in Sec.~\ref{sec:MARL}, we are interested in the following question. Given an $\alpha >
0$ and $p > 0$, how many samples~$n$ are needed from the generative model to ensure that
$\approxpol_n$ is an $\alpha$-optimal for $\MDP$ with at least probability $1 - p$. This is called the \emph{sample
complexity} of learning.  

We state two bounds on the sample complexity of generative models: one based on Hoeffding inequality and the other based on Bernstein inequality. For the first bound, we provide a complete proof to provide an exact characterization of the constants.

Suppose $X \in \mathcal{X}$ is a random variable with distribution~$\mu$. Suppose
$\{X_1, \dots, \allowbreak X_n\}$ is a sequence of random variables sampled according
to~$\mu$. Let $\hat \mu_n$ denote the empirical measure constructed from
$\{X_1, \dots, X_n \}$, i.e., 
\[
  \hat \mu_n(x) \coloneqq \frac1n \sum_{k=1}^n \IND_{\{X_k = x\}},
\]
where $\IND_{\{E\}}$ denotes the indicator function of the event~$E$.
Then, we have the following.

\begin{lemma}\label{lem:Hoeffding}
  For a given $H > 0$, let $\F_H$ denote the set of functions $f
  \colon \ALPHABET X \to \reals$ such that $\Span(f) \le H$. Then, for any $f
  \in \F_H$ and $\Delta > 0$,
  \[
    \PR\biggl( \biggl| \sum_{x \in \ALPHABET X}\bigl[ \mu(x) f(x) - \mu_n(x) f(x) \bigr] \biggr|
    \ge \Delta \biggr) \le \\exp{\left(-\frac{2n\Delta^2}{H^2}\right)}.
  \]
\end{lemma}
\begin{proof}
  Let $Z \coloneqq f(X)$ and $Z_i \coloneqq f(X_i)$. Then, $\{Z_1, \dots, Z_n\}$ is an i.i.d.\@ sequence and
  $\Supp(Z_i) \le H$. Then, by the Hoeffding inequality \citep[Corollary A.1]{cesa2006prediction}, 
  \[
    \PR\biggl( \biggl| \frac1n \sum_{i=1}^n Z_i -  \EXP[Z] \biggr|
    \ge \Delta \biggr) \le 2\exp{\left(-\frac{2n\Delta^2}{H^2}\right)}.
  \]
  The result then follows from observing that $\EXP[Z] = \sum_{x \in \ALPHABET
  X} \mu(x) f(x)$ and $(\sum_{i=1}^n Z_i)/n = \sum_{x \in \ALPHABET X}
  \mu_n(x) f(x)$. \qed
\end{proof}

\begin{theorem}[Hoeffding-Type Bound]\label{thm:generative}
  For any $\alpha_{\mathrm{opt}} > 0$, $\alpha > (1+\discount)\alpha_{\mathrm{opt}}/{(1-\discount)}$ and $p > 0$, let
  
\begin{align}\label{eq:nH:MDP}
    n \ge \biggl\lceil \biggl( \frac{ \discount }{(1 - \discount)^2 } \Span(r) \biggr)^2 
    \frac{ 2 \log(2 |\StSp|\,|\ActSp|p^{-1}) }{(\alpha - \frac{1+\discount}{1-\discount}\alpha_{\mathrm{opt}})^2} \biggr\rceil.
\end{align}
  Then, with probability $1 - p$, the $\alpha_{\mathrm{opt}}$-optimal strategy $\approxpol_n$ of MDP
  $\approxMDP_n$ is $\alpha$-optimal for $\MDP$.
\end{theorem}

\begin{proof}
 The proof follows from the application of Lemma~\ref{lem:Hoeffding} and Theorem~\ref{cor:MDP-approx}.

  Let $H \coloneqq \Span(r)$, $\pol_{*}$ be an optimal policy of $\MDP$ and $\Valf_{\pol_{*}}$ be the optimal value functions for $\MDP$. From Lemma~\ref{lem:v-span}, we know that $\Valf_{\pol_{*}} \in \F_H$.
  Therefore, from Lemma~\ref{lem:Hoeffding}, we have that for a given state-action pair
  $(\st,\act)$ and $\Delta > 0$, 
  \[
    \PR\biggl( \biggl\{\biggl|
      \sum_{\st' \in \StSp} \Trans(\st' | \st, \act) \Valf_{\pol_{*}}(\st')
      -
      \sum_{\st' \in \StSp} \approxTrans_n(\st' | \st, \act) \Valf_{\pol_{*}}(\st')
    \biggr| \ge \Delta \biggr\}\biggr)
    \le 2\exp{\left(-\frac{2n\Delta^2}{H^2}\right)}.
  \]

  Therefore, by
  the union bound,
  \begin{align}
    \PR\biggl(& 
     \biggl\{ \max_{(\st,\act) \in \StSp \times \ActSp} \biggl|
      \sum_{\st' \in \StSp} \Trans(\st' | \st, \act) \Valf_{\pol_{*}}(\st')
      -
      \sum_{\st' \in \StSp} \approxTrans_n(\st' | \st, \act) \Valf_{\pol_{*}}(\st')
    \biggr| \ge \Delta \biggr\} \biggr)
    \notag \\
    &\le 2 |\StSp|\,|\ActSp|
     \exp{\left(-\frac{2n\Delta^2}{H^2}\right)}.
     \label{eq:union-bound}
  \end{align}
  Now, choose $\Delta$ such that the right hand side of~\eqref{eq:union-bound}
  equals $p$, i.e.,
  \[
    \Delta \coloneqq H \sqrt{
      \frac{\log(2 |\StSp|\,|\ActSp|p^{-1})}
      {2n}.
    }
  \]
  Theorem~\ref{cor:MDP-approx} implies that $\hat \pol_n$ is an $\alpha$-optimal policy for MDP $\ALPHABET{M}$, where 
  \[
     \alpha \le \frac{2\discount}{(1 - \discount)^2} H \sqrt{\frac{\log(2 |\StSp|\,|\ActSp|p^{-1})} {2n}} + \frac{1+\discount}{1-\discount}\alpha_{\text{opt}}
  \]
  The result follows from substituting the value of $n$ from \eqref{eq:nH:MDP} in the above equation.
  \qed
\end{proof}

\begin{remark}\label{rem:tight}
  The sample complexity bound in Theorem~\ref{thm:generative} is not tight. It is shown in \cite{azar2013minimax} that finding an $\alpha$-optimal policy with probability $1-p$ requires at least
  \[
    \Omega\biggl(|\StSp|\,|\ActSp|\frac{\log(|\StSp|\,|\ActSp|p^{-1})}{(1-\discount)\alpha^2} \biggr)
  \]
  samples.\footnote{Recall that we are working with normalized total expected reward (see Remark~\ref{rem:normalize}), while the results \cite{azar2013minimax} are derived for the unnormalized total reward. In the discussion above, we have normalized the results of \cite{azar2013minimax}.} The upper bound in Theorem~\ref{thm:generative} is loose by a factor of $1/(1-\discount)^3$. For the case of MDPs, tighter upper bounds which match the lower bound of~\cite{azar2013minimax} (up to logarithmic factors) have been obtained in~\cite{sidford2018near,agarwal2020model} by using the Bernstein inequality rather than the Hoeffding inequality. We present these bounds below.
\end{remark}

\begin{theorem}\label{thm:generative-kakade}
For any $\alpha_{\mathrm{opt}} > 0$, $\alpha \in (5\alpha_{\mathrm{opt}}, \|r\|_{\infty}\sqrt{(1 - \discount)} + 5 \alpha_{\mathrm{opt}})$ and $p > 0$, let
  \[
    n \ge \frac{c\discount \|r\|_{\infty}^2\log(8|\StSp||\ActSp|(1-\discount)^{-2}p^{-1})}{(1-\discount)(\alpha - 5\alpha_{\mathrm{opt}})^2} ,
  \]
  where $c$ is an absolute constant. Then, with probability $1 - p$, any $\alpha_{\mathrm{opt}}$-optimal strategy $\approxpol_n$ of MDP
  $\approxMDP_n$ is $\alpha$-optimal for $\MDP$.
\end{theorem}
\begin{proof}
 The proof follows from \cite[Theorem 1]{agarwal2020model} with appropriate scaling to account for normalization of rewards. However, the expression for $n$ above differs slightly from the expression in~\cite[Theorem 1]{agarwal2020model}. In particular, let $L = \log(8|\StSp| |\ActSp| (1-\discount)^{-1} p^{-1})$ and $L' = \log(8|\StSp| |\ActSp| (1-\discount)^{-2} p^{-1})$ (the exponent of $(1-\discount)$ is different in the two expressions). The expression in~\cite[Theorem 1]{agarwal2020model} (rescaled for normalized rewards) states that
 \[
    n \ge \frac{c \discount \| r\|_{\infty}^2 L } { (1-\discount) \alpha^2 },
\]
while the expression in the statement of Theorem~\ref{thm:generative-kakade} states that
 \[
    n \ge \frac{c \discount \| r\|_{\infty}^2 L' } { (1-\discount) \alpha^2 }.
\]
The reason for this difference is that there is a typo in  \cite[Lemma 11]{agarwal2020model}, which is carried over in all the results. In particular, in the proof of \cite[Lemma 11]{agarwal2020model} it is claimed that for $|U_s| = 1/(1-\discount)^2$ and $p' = p/(2|\StSp||\ActSp|)$ (in \cite{agarwal2020model}, the symbol $\delta$ is used instead of $p$), we have $\log(4|U_s|/p') = L$ but elementary algebra shows that $\log(4|U_s|/p') = L'$. 
\endgraf
Also note that in \cite{agarwal2020model}, the bound of Theorem~\ref{thm:MARL2} was simplified as 
  \[
    n \ge \frac{c\discount \|r\|_{\infty}^2\log(|\StSp||\ActSp|(1-\discount)^{-2}p^{-1})}{(1-\discount)(\alpha - 5\alpha_{\mathrm{opt}})^2} ,
  \]
  i.e., the multiplicative factor of $8$ inside the log was removed. Since we want to compare the Bernstein-type bound with the Hoeffding-type bound in Sec.~\ref{sec:comparison}, we carry the multiplicative factor of $8$ in our expression.
\qed
\end{proof}

\section{Proof of the main results}\label{sec:proof}

\subsection{Road map of the proof}

As mentioned in the beginning of Sec.~\ref{sec:MDP}, the main idea of the proofs is to look at the best response of a player to the pre-specified strategy profile of other players. To formally establish that the problem of finding the best response is an MDP, we first start by characterizing the best response in terms of Bellman operators and stating existing results that characterize MPE and approximate MPE in terms of Bellman operators (Propositions~\ref{prop:MPE-charac} and \ref{prop:approxMPE-charac}). Then, we formally define a ``best response MDP'' and show that MPE and approximate MPE can be stated in terms of such ``best response MDPs'' (Corollaries~\ref{cor:MPE-MDP} and \ref{cor:alpha-MPE-MDP}).

To establish the robustness results, we show that the ``best response MDP'' corresponding to an $(\varepsilon, \delta)$-approximation of a game is an $(\varepsilon, \delta)$-approximation of the ``best response MDP'' of the original game (Lemma~\ref{lem:MDPs-relation}). This allows us to generalize the approximation results of MDPs to games. We then build on this relationship to generalize the sample complexity results of MDPs to games.

\subsection{Bellman operators and characterization of Markov perfect equilibrium}
Given a Markov strategy profile $\pol \coloneqq \polAll$, state $\st \in \StSp$, and action profile $\act \coloneqq (\act^i)_{i \in \playerset} \in \ActSp$, we use the notation
\begin{align}
    \pol(\act|\st) &\coloneqq \prod_{i \in \playerset}\pol^i(\act^i|\st) \quad \text{and} \notag \\
    \pol^{-i}(\act^{-i}|\st) &\coloneqq \prod_{j \in \playerset \setminus \{i\}}\pol^j(\act^j|\st).
\end{align}

Given a player $i \in \playerset$ and a Markov strategy profile $\pol \coloneqq (\pol^i, \pol^{-i})$, we define two Bellman operators as follows:
\begin{enumerate}
  \item An operator $\Bellman^i_{(\pol^i, \pol^{-i})}: \reals^{|\StSp|} \to \reals^{|\StSp|}$ given as follows: for any $\valf \in \reals^{|\StSp|}$ and $\st \in \StSp$,
    \begin{equation*}
        [\Bellman^i_{(\pol^i, \pol^{-i})}\valf](\st)  
        \coloneqq \sum_{\act \in \ActSp}\pol(\act|\st)\Bigl[(1-\discount)\rew^i(\st, \act) + \discount\sum_{\st' \in \StSp}\Trans(\st'|\st, \act)\valf(\st')\Bigr].
    \end{equation*}
  \item An operator $\Bellman^i_{*,\pol^{-i}}: \reals^{|\StSp|} \to \reals^{|\StSp|}$ given as follows: for any $\valf \in \reals^{|\StSp|}$ and $\st \in \StSp$,
    \begin{multline*}
        [\Bellman^i_{(*, \pol^{-i})}\valf](\st) 
        \coloneqq \max_{\act^i \in \ActSp^i}\Biggl[\sum_{\act^{-i} \in \ActSp^{-i}}\pol^{-i}(\act^{-i}|\st)\Bigl[(1-\discount)\rew^i(\st, \act) \\
        + \discount\sum_{\st' \in \StSp}\Trans(\st'|\st, \act)\valf(\st')\Bigr]\Biggr].
    \end{multline*}
\end{enumerate}

Now, MPE and approximate MPE can be characterized using the Bellman operators. These are standard results. See, for example, \cite{FilarCompetitive1996}.

\begin{proposition}\label{prop:MPE-charac}
A Markov strategy profile $\pol \coloneqq \polAll$ is an MPE if and only if there exist \textbf{value functions} $\Valf^i \in \reals^{|\StSp|}, i \in \playerset$, such that
\begin{equation}
    \Valf^i = \Bellman^i_{(\pol^i, \pol^{-i})}\Valf^i
    \quad\text{and}\quad
    \Valf^i = \Bellman^i_{(*, \pol^{-i})}\Valf^i,  \quad \forall i \in \playerset.
\end{equation}
\end{proposition}

An immediate consequence of Proposition~\ref{prop:MPE-charac} and the definition of approximation MPE is the following.

\begin{proposition}\label{prop:approxMPE-charac}
  Given a Markov strategy profile $\pol \coloneqq \polAll$, for any $i \in
  \playerset$, let $\Valf_\pol^i$ be the unique fixed point of 
  \(
    \Valf^i_\pol = \Bellman^i_{(\pol^i, \pol^{-i})}\Valf^i_\pol
  \) and let $\Valf_{(*,\pol^{-i})}^i$ be the unique fixed point of
  \(
    \Valf^i_{(*,\pol^{-i})} = \Bellman^i_{(*, \pol^{-i})}\Valf^i_{(*,\pol^{-i})}
  \).
  Then, the strategy profile $\pol$ is an $\alpha$-MPE, 
  $\alpha  \coloneqq (\alpha^i)_{i \in\playerset}$, 
  if and only if
  \begin{equation}
    \Valf_\pol^i \ge \Valf_{(*,\pol^{-i})}^i - \alpha^i \ONES, 
    \quad \forall i \in \playerset.
  \end{equation}
\end{proposition}
\begin{proof}
  The proof follows from arguments similar to the proof of Prop.~\ref{prop:approx-optimal-strategy-MDP}.
\end{proof}
\subsection{Relationship between games and MDPs}
Given a game $\game \coloneqq  \langle \playerset, \StSp, \ActSpAll, \Trans, \rewAll, \discount \rangle$ and a Markov strategy \linebreak $\pol \coloneqq \polAll$, we can define MDPs $\{\MDP^i_{\pol^{-i}}\}_{i \in \playerset}$ as follows. For player $i \in \playerset$, MDP  $\MDP^i_{\pol^{-i}} \coloneqq \langle \StSp, \ActSp^i, \Trans^i_{\pol^{-i}}, \rew^i_{\pol^{-i}}, \discount \rangle$, where the transition matrix $\Trans^i_{\pol^{-i}}: \StSp \times \ActSp^i \to \PSP(\StSp)$ is given by
\begin{equation}\label{eq:MDP-Trans}
    \Trans^i_{\pol^{-i}}(\st'|\st, \act^i) \coloneqq  \sum_{\act^{-i} \in \ActSp^{-i}}\pol^{-i}(\act^{-i}|\st)\Trans(\st'|\st, (\act^i, \act^{-i})),
\end{equation}
and the reward function $\rew^i_{\pol^{-i}}: \StSp \times \ActSp^i \to \reals$ is given by
\begin{equation}\label{eq:MDP-rew}
    \rew^i_{\pol^{-i}}(\st, \act^i) \coloneqq \sum_{\act^{-i} \in \ActSp^{-i}}\pol^{-i}(\act^{-i}|\st)\rew^i(\st, (\act^i, \act^{-i})).
\end{equation}
In other words, in $\MDP^i_{\pol^{-i}}$ , the strategy of player $i$ may chosen freely while the strategies of all other players are fixed at those specified in $\pol^{-i}$. Note the Bellman operators $\Bellman^i_{(\pol^i, \pol^{-i})}$ and $\Bellman^i_{(*, \pol^{-i})}$ corresponding to game $\game$ and strategy $\pol$ are the same as Bellman operators of MDP $\MDP^i_{\pol^{-i}}$.
Therefore, by combining Propositions~\ref{prop:optimal-strategy-MDP} and~\ref{prop:MPE-charac}, we have the following:
\begin{corollary}\label{cor:MPE-MDP}
  A Markov strategy profile $\pol \coloneqq \polAll$ is an MPE if and only if for every $i \in \playerset$, the strategy $\pol^i$ is an optimal strategy for MDP $\MDP^i_{\pol^{-i}}$.
\end{corollary}
\begin{proof}

 This is an immediate consequence of the definition of $\MDP^i_{\pol^{-i}}$. For the sake of completeness, we provide a formal proof. 
  Arbitrarily pick a player $i \in \playerset$ and consider any Markov policy $\tilde \pol^i$ for player~$i$. The Bellman operator $\Bellman^i_{(\tilde \pol^i, \pol^{-i})}$ of game~$\game$ is the same as the Bellman operator for evaluating policy $\tilde \pol^i$ in MDP~$\MDP^i_{\pol^{-i}}$. Thus, the value function $\Valf_{\tilde \pol^i, \pol^{-i}}$ (which is the fixed point of $\Bellman^i_{(\tilde \pol^i, \pol^{-i})}$) is equal to the value of policy $\tilde \pol^i$ in MDP~$\MDP^i_{\pol^{-i}}$. Now,  we prove the two directions separately.
  \begin{enumerate}
      \item[($\Rightarrow$)] Suppose $\pol$ is an MPE of game $\game$. By the definition of MPE, for any player $i \in \playerset$ and any policy $\tilde \pol^i$, we have $\Valf_{(\pol^i, \pol^{-i})}(\st) \ge \Valf_{(\tilde \pol^i, \pol^{-i})}(\st)$, for all $\st \in \StSp$. This means that in MDP~$\MDP^i_{\pol^{-i}}$, the performance of policy $\pol^i$ is at least as good as the performance of any other policy $\tilde \pol^i$. Hence, policy $\pol^i$ is optimal for MDP~$\MDP^i_{\pol^{-i}}$. 
      \item[$(\Leftarrow$)] Suppose for all player~$i \in \playerset$, the policy $\pol^i$ is optimal for MDP~$\MDP^i_{\pol^{-i}}$. This means that for any other policy $\tilde \pol^i$ for player~$i$, the performance of policy $\pol^i$ in MDP~$\MDP^i_{\pol^{-i}}$ is at least as good as the performance of policy $\tilde \pol^i$. Thus, we have $\Valf_{(\pol^i, \pol^{-i})}(\st) \ge \Valf_{(\tilde \pol^i, \pol^{-i})}(\st)$, for all $\st \in \StSp$. Since this is true for every player~$i \in \playerset$, the policy $\pol$ is an MPE. 
      \qed
  \end{enumerate}

\end{proof}

Similarly, by combining Propositions~\ref{prop:approx-optimal-strategy-MDP} and~\ref{prop:approxMPE-charac}, we have the following:
\begin{corollary}\label{cor:alpha-MPE-MDP}
  Given approximate levels $\alpha \coloneqq (\alpha^i)_{i \in \playerset}$, $\alpha^i \in \reals_{\ge 0}$, a Markov strategy profile $\pol \coloneqq \polAll$, is an $\alpha$-MPE if and only if for every $i \in \playerset$, the strategy $\pol^i$ is an $\alpha^i$-optimal strategy for MDP $\MDP^i_{\pol^{-i}}$.
\end{corollary}
\begin{proof}

The proof argument is almost the same as the proof of Corollary~\ref{cor:MPE-MDP}.
  As argued in the proof of Corollary~\ref{cor:MPE-MDP}, the value function $\Valf_{\tilde \pol^i, \pol^{-i}}$ (which is the fixed point of $\Bellman^i_{(\tilde \pol^i, \pol^{-i})}$) is equal to the value of policy $\tilde \pol^i$ in MDP~$\MDP^i_{\pol^{-i}}$. Now,  we prove the two directions separately.
  \begin{enumerate}
      \item[($\Rightarrow$)] Suppose $\pol$ is an $\alpha$-MPE of game $\game$. By the definition of MPE, for any player $i \in \playerset$ and any policy $\tilde \pol^i$, we have $\Valf_{(\pol^i, \pol^{-i})}(\st) \ge \Valf_{(\tilde \pol^i, \pol^{-i})}(\st) - \alpha^i$, for all $\st \in \StSp$. This means that in MDP~$\MDP^i_{\pol^{-i}}$, the performance of policy $\pol^i$ is at least as good as the performance of any other policy $\tilde \pol^i$ minus $\alpha^i$. Hence, policy $\pol^i$ is $\alpha^i$-optimal for MDP~$\MDP^i_{\pol^{-i}}$. 
      \item[$(\Leftarrow$)] Suppose for all player~$i \in \playerset$, the policy $\pol^i$ is $\alpha^i$-optimal for MDP~$\MDP^i_{\pol^{-i}}$. This means that for any other policy $\tilde \pol^i$ for player~$i$, the performance of policy $\pol^i$ in MDP~$\MDP^i_{\pol^{-i}}$ is at least as good as the performance of policy $\tilde \pol^i$ minus $\alpha^i$. Thus, we have $\Valf_{(\pol^i, \pol^{-i})}(\st) \ge \Valf_{(\tilde \pol^i, \pol^{-i})}(\st) - \alpha^i$, for all $\st \in \StSp$. Since this is true for every player~$i \in \playerset$, the policy $\pol$ is an $\alpha$-MPE. \qed
  \end{enumerate}

\end{proof}
\subsection{Relationship between MDPs corresponding to a strategy profile}
We first provide a preliminary result.
\begin{lemma}\label{lem:Delta-bound}
  For any function $f \colon \StSp \to \reals$, transitions $\Trans,
    \arbitraryTrans : \StSp \times \ActSpAll \to
  \Delta(\StSp)$, player~$i \in \playerset$, strategy $\pol^{-i}$ for
  players other than~$i$, $(\st, \act^i) \in \StSp \times \ActSp^i$ and transitions $\Trans_{\pol^{-i}}, \arbitraryTrans_{\pol^{-i}}: \StSp^i \times \ActSp^i \to \Delta(\StSp^i)$ defined as in~\eqref{eq:MDP-Trans},  we have
  \begin{align*}
    \hskip 2em & \hskip -2em
 \biggl\lvert \sum_{\st' \in \StSp} f(\st')\Trans^i_{\pol^{-i}}(\st'|\st,\act^i) 
      - \sum_{\st' \in \StSp} f(\st')\arbitraryTrans^i_{\pol^{-i}}(\st'|\st,\act^i) \biggr\rvert 
    \\
    & \le
    \max_{\act^{-i} \in \ActSp^{-i}} 
    \biggl\lvert \sum_{\st' \in \StSp} f(\st') \Trans(\st'|\st,(\act^i, \act^{-i})) -\sum_{\st' \in \StSp} f(\st')  \arbitraryTrans(\st'|\st,(\act^i, \act^{-i})) \biggr\lvert.
  \end{align*}
  Therefore,
  \begin{align*}
    \hskip 2em & \hskip -2em
    \max_{\st \in \StSp, \act^i \in \ActSp^i}
 \biggl\lvert \sum_{\st' \in \StSp} f(\st')\Trans^i_{\pol^{-i}}(\st'|\st,\act^i) 
      - \sum_{\st' \in \StSp} f(\st')\arbitraryTrans^i_{\pol^{-i}}(\st'|\st,\act^i) \biggr\rvert 
    \\
    & \le
      \max_{\st \in \StSp, (\act^i, \act^{-i}) \in \ActSp} 
      \biggl\lvert \sum_{\st' \in \StSp} f(\st') \Trans(\st'|\st,(\act^i, \act^{-i})) - \sum_{\st' \in \StSp} f(\st') \arbitraryTrans(\st'|\st,(\act^i, \act^{-i})) \biggr\lvert.
  \end{align*}
\end{lemma}
\begin{proof}
  For the first part, from definition of $\arbitraryTrans^i_{\pol^{-i}}$, we have
  \begin{align*}
    \hskip 2em & \hskip -2em
      \biggl\lvert \sum_{\st' \in \StSp} f(\st')\Trans^i_{\pol^{-i}}(\st'|\st,\act^i) 
      - \sum_{\st' \in \StSp} f(\st')\arbitraryTrans^i_{\pol^{-i}}(\st'|\st,\act^i) \biggr\rvert 
      \displaybreak[1] \notag \\
      &=   
      \biggl\lvert \sum_{\st' \in \StSp} \sum_{\act^{-i} \in \ActSp^{-i}} f(\st')\pol^{-i}(\act^{-i}|\st)\Trans(\st'|\st,(\act^i, \act^{-i})) \notag \\ 
      & \hskip 6em - \sum_{\st' \in \StSp}\sum_{\act^{-i} \in \ActSp^{-i}} f(\st')\pol^{-i}(\act^{-i}|\st)\arbitraryTrans(\st'|\st,(\act^i, \act^{-i})) \biggr\rvert 
      \displaybreak[1] \notag \\
      & \le   \biggl\lvert \sum_{\act^{-i} \in \ActSp^{-i}} \pol^{-i}(\act^{-i}|\st)  \notag \\
      & \qquad \times \biggl[ \sum_{\st' \in \StSp} f(\st') (\Trans(\st'|\st,(\act^i, \act^{-i})) - \arbitraryTrans(\st'|\st,(\act^i, \act^{-i}))) \biggr] \biggr\lvert \displaybreak[1] \notag \\
      & \le   \sum_{\act^{-i} \in \ActSp^{-i}} \pol^{-i}(\act^{-i}|\st) \notag \\
      & \qquad \times \biggl\lvert \sum_{\st' \in \StSp} f(\st') (\Trans(\st'|\st,(\act^i, \act^{-i})) - \arbitraryTrans(\st'|\st,(\act^i, \act^{-i}))) \biggr\lvert 
      \displaybreak[1] \notag \\
      & \le\sum_{\act^{-i} \in \ActSp^{-i}} \pol^{-i}(\act^{-i}|\st) \notag \\
      & \quad \times \max_{\tilde \act^{-i} \in \ActSp^{-i}} 
      \biggl\lvert \sum_{\st' \in \StSp} f(\st') (\Trans(\st'|\st,(\act^i, \tilde \act^{-i})) - \arbitraryTrans(\st'|\st,(\act^i, \tilde \act^{-i}))) \biggr\lvert 
      \notag \\
      &= \max_{ \tilde \act^{-i} \in \ActSp^{-i}} 
      \biggl\lvert \sum_{\st' \in \StSp} f(\st') (\Trans(\st'|\st,(\act^i, \tilde \act^{-i})) - \arbitraryTrans(\st'|\st,(\act^i, \tilde \act^{-i}))) \biggr\lvert .
  \end{align*}
  The second part following by taking a maximum over $(\st, \act^i)$.
  \qed
\end{proof}

Suppose we are given a game $\game$ and its $\aiserror$ approximation $\approxgame$. Moreover, suppose $\approxpol \coloneqq \polAllApprox$ is an MPE of $\approxgame$.

Let $\{\approxMDP^i_{\approxpol^{-i}}\}$ be the MDPs corresponding to game $\approxgame$ and strategy $\approxpol$. Similarly, let $\{\MDP^i_{\approxpol^{-i}}\}$ be the MDPs corresponding to game $\game$ and strategy $\approxpol$. An immediate implication of Lemma~\ref{lem:Delta-bound} is the following.
\begin{lemma}\label{lem:MDPs-relation}
For any player $i \in \playerset$, MDP $\approxMDP^i_{\approxpol^{-i}}$ is an $\aiserror$ approximation of MDP $\MDP^i_{\approxpol^{-i}}$.
\end{lemma}
\begin{proof}
  Consider 
  \begin{align}\label{eq:rew-MDP-relation}
    \hskip 2em & \hskip -2em 
      \lvert \rew^i_{\approxpol^{-i}}(\st, \act^i) - \approxrew^i_{\approxpol^{-i}}(\st, \act^i) \rvert \notag \\ &\stackrel{(a)}{\le} \sum_{\act^{-i} \in \ActSp^{-i}} \approxpol^{-i}(\act^{-i}|\st) \lvert \rew^i(\st, (\act^i, \act^{-i})) - \approxrew^i(\st, (\act^i, \act^{-i})) \rvert \notag \\
      & \stackrel{(b)}{\le} \sum_{\act^{-i} \in \ActSp^{-i}} \approxpol^{-i}(\act^{-i}|\st) \varepsilon \notag \\
      & \stackrel{(c)}{=} \varepsilon,
  \end{align}
  where $(a)$ follows from~\eqref{eq:MDP-rew}, $(b)$ follows from~\eqref{eq:game-rew-approx} and $(c)$ follows as $\varepsilon$ is independent of $\act^{-i}$.
  Furthermore,
  \begin{align}\label{eq:Trans-MDP-relation}
    \hskip 1em & \hskip -1em
       \max_{\st \in \StSp, \act^i \in \ActSp^i}  d_\F(\Trans^i_{\approxpol^{-i}}(\cdot|\st, \act^i),\approxTrans^i_{\approxpol^{-i}}(\cdot|\st, \act^i)) \notag \\
       &\stackrel{(d)}{=}\sup_{f \in \F}  \max_{\st \in \StSp, \act^i \in \ActSp^i} \biggl\lvert \sum_{\st' \in \StSp} f(\st')\Trans^i_{\approxpol^{-i}}(\st'|\st,\act^i) 
      - \sum_{\st' \in \StSp} f(\st')\approxTrans^i_{\approxpol^{-i}}(\st'|\st,\act^i) \biggr\rvert 
      \displaybreak[1] \notag \\
      & \stackrel{(e)}{\le}\sup_{f \in \F}  \max_{\substack{\st \in \StSp \\ (\act^i, \act^{-i}) \in \ActSp}}
      \biggl\lvert \sum_{\st' \in \StSp} f(\st') (\Trans(\st'|\st,(\act^i, \act^{-i})) - \approxTrans(\st'|\st,(\act^i, \act^{-i}))) \biggr\lvert 
      \displaybreak[1] \notag \\
      & \stackrel{(f)}= \max_{\substack{\st \in \StSp \\ (\act^i, \act^{-i}) \in \ActSp}} 
      d_\F(\Trans(\cdot|\st, (\act^i, \act^{-i})),\approxTrans(\cdot|\st, (\act^i,\act^{-i}))) \notag \\
      & \stackrel{(g)}{=}\delta
  \end{align}
  where $(d)$ and $(f)$ follows from Definition~\ref{def:ipm}, $(e)$ follows from Lemma~\ref{lem:Delta-bound}, and $(g)$ follows from Definition~\ref{def:approx-game}.
  
  Equations~\eqref{eq:rew-MDP-relation} and \eqref{eq:Trans-MDP-relation} imply that MDP $\approxMDP^i_{\approxpol^{-i}}$ is an $\aiserror$-approximation of $\MDP^i_{\approxpol^{-i}}$ (see Definition~\ref{def:approxMDP}).
  \qed
\end{proof}

\subsection{Proof of Theorem~\ref{thm:main}}
Arbitrarily fix a player $i \in \playerset$. Then, we have the following.
\begin{enumerate}
    \item From Corollary~\ref{cor:alpha-MPE-MDP}, since $\approxpol$ is an $\alpha_{\mathrm{opt}}$-MPE of $\approxgame$, we have that the strategy $\approxpol^i$ is $\alpha_{\mathrm{opt}}$-optimal for MDP $\approxMDP^i_{\approxpol^{-i}}$.
    \item From Lemma~\ref{lem:MDPs-relation}, we know that MDP $\approxMDP^i_{\approxpol^{-i}}$ is an $\aiserror$ approximation of MDP $\MDP^i_{\approxpol^{-i}}$. Then, by Theorem~\ref{thm:MDP-approx}, we get that strategy $\approxpol^i$ is an $\alpha^i$-optimal strategy for MDP $\MDP^i_{\approxpol^{-i}}$, where $\alpha^i$ is given by Theorem~\ref{thm:main}. Lemma~\ref{lem:Delta-bound} shows that
    \[
    \max_{\st \in \StSp, \act^i \in \ActSp^i}
 \biggl\lvert \sum_{\st' \in \StSp} \approxValf^{i}_{(\approxpol^i, \approxpol^{-i})}(\st') \Trans^i_{\approxpol^{-i}}(\st'|\st,\act^i) 
      - \sum_{\st' \in \StSp} \approxValf^{i}_{(\approxpol^i, \approxpol^{-i})}(\st') \approxTrans^i_{\approxpol^{-i}}(\st'|\st,\act^i) \biggr\rvert
    \]
      is upper bounded by $\Delta^i_{(\approxpol^i,\approxpol^{-i})}$ given in Theorem~\ref{thm:main}.
      By a similar argument, we can show that
    \[
    \max_{\st \in \StSp, \act^i \in \ActSp^i}
 \biggl\lvert \sum_{\st' \in \StSp} \approxValf^{i}_{(\tilde \pol^i_*, \approxpol^{-i})}(\st') \Trans^i_{\approxpol^{-i}}(\st'|\st,\act^i) 
      - \sum_{\st' \in \StSp} \approxValf^{i}_{(\tilde \pol^i_*, \approxpol^{-i})}(\st') \approxTrans^i_{\approxpol^{-i}}(\st'|\st,\act^i) \biggr\rvert
    \]
      is upper bounded by $\Delta^i_{(\tilde \pol^i_*,\approxpol^{-i})}$ given in Theorem~\ref{thm:main}.
    \item Since the above results hold for all $i \in \playerset$, Corollary~\ref{cor:alpha-MPE-MDP} implies that strategy profile $\approxpol$ is an $\alpha$-MPE of $\game$, where $\alpha \coloneqq (\alpha^i)_{i \in \playerset}$ and $\alpha^i$ is given by Theorem~\ref{thm:main}.
    \item The specific formulas for $\alpha$ in Corollaries~\ref{cor:TV} and~\ref{cor:W} follow from Corollaries~\ref{cor:ais-tv} and \ref{cor:ais-w}.
\end{enumerate}

\subsection{Proofs of Theorem~\ref{thm:MARL} and Theorem~\ref{thm:MARL2}}
The proof argument for Theorems~\ref{thm:MARL} and~\ref{thm:MARL2} are similar, so we prove them together. Theorems~\ref{thm:generative} and~\ref{thm:generative-kakade} show that for any $p > 0$, $\alpha_{\mathrm{opt}} > 0$ and $\alpha > \frac{1+\discount}{1-\discount}\alpha_{\mathrm{opt}}$ which satisfies appropriate conditions, there exists a function $N(p)$, such that for $n \ge N(p)$ any
$\alpha_{\mathrm{opt}}$-optimal policy $\approxpol_n$ of MDP $\approxMDP_n$ is an $\alpha$-optimal for MDP $\MDP$.

Now consider the approximate game $\approxgame_n$ and let $\pol$ be an $\alpha_{\mathrm{opt}}$-MPE of game $\approxgame$. For any player~$i \in \playerset$, let $\mathcal{E}^i_p$ denote the event that policy $\pol^i$ is \emph{not} $\alpha$-optimal for $\MDP^i_{\pol^{-i}}$. The above restatement of Theorems~\ref{thm:generative} and~\ref{thm:generative-kakade} imply that for $n \ge N(p/|\playerset|)$, 
\(
  \PR(\mathcal{E}^i_{p/|\playerset|}) \le p/|\playerset|.
\)
Therefore, by the union bound, 
\[
  \PR\biggl( \bigcup_{i \in \playerset} \mathcal{E}^i_{p/|\playerset|} \biggr) 
  \le \sum_{i \in \playerset} \PR(\mathcal{E}^i_{p/|\playerset|})
  \le p.
\]
Thus, for $n \ge N(p/|\playerset|)$, the policy $\pol^i$ is $\alpha$-optimal for MDP~$\MDP^i_{\pol^{-i}}$, for all $i \in \playerset$. Thus, by Corollary~\ref{cor:alpha-MPE-MDP}, $\pol$ is an $\alpha$-MPE of game~$\game$. 

The results of Theorems~\ref{thm:MARL} and~\ref{thm:MARL2} then follow by substituting the expressions for $N(p/|\playerset|)$ from Theorems~\ref{thm:generative} and~\ref{thm:generative-kakade}, respectively.

\section{Numerical examples} \label{sec:numerical_examples}
In this section, we present two numerical examples to demonstrate the main results of Theorem~\ref{thm:main} and~\ref{thm:MARL}.
\subsection{Robustness of Markov perfect equilibrium}\label{sec:example}
Consider a setting where $\playerset=\{1,2\}$, $\StSp = \{1,2,3\}$, $\ActSp_1 = \ActSp_2 = \{1,2\}$, and $\discount = 0.9$. We consider two games: original game $\game$ and approximate game $\approxgame$ which differ in their reward functions and transition matrices. We describe the transition matrices as $\{\Trans(\act)\}_{\act \in \ActSp}$, where $\Trans(\act) = [\Trans(\st' \mid \st, \act)]_{\st, \st' \in \StSp}$ and describe the reward functions as $\{\rew(\st)\}_{\st \in \StSp}$ where $\rew(\st)$ is the bi-matrix $[(\rew_1(\st, (\act_1, \act_2)),\rew_2(\st, (\act_1, \act_2)))]_{(\act_1, \act_2) \in \ActSp}$. 

For the original game $\game$, we have
\begin{gather*}
\rew(1) = \begin{bmatrix}
  (1.0, 0.4) & (0.7, 1.0)\\
  (0.3, 1.0) & (0.8, 0.7)
\end{bmatrix}, 
\quad
\rew(2) = \begin{bmatrix}
  (0.6, 0.7) & (0.7, 0.6)\\
  (0.3, 0.8) & (0.2, 0.2)
\end{bmatrix} ,
\\
\rew(3) = \begin{bmatrix}
  (0.2, 0.6) & (0.1, 0.7)\\
  (0.6, 0.7) & (0.5, 0.3)
\end{bmatrix},
\end{gather*}
and
\begin{align*}
\Trans((1,1)) &= \begin{bmatrix}
0.40 & 0.40 & 0.20\\
0.10 & 0.50 & 0.40\\
0.40 & 0.10 & 0.50
\end{bmatrix},
&
\Trans((1,2)) &= \begin{bmatrix}
0.30 & 0.40 & 0.30\\
0.20 & 0.20 & 0.60\\
0.30 & 0.35 & 0.35
\end{bmatrix},
\\
\Trans((2,1)) &= \begin{bmatrix}
0.25 & 0.25 & 0.50\\
0.30 & 0.30 & 0.40\\
0.20 & 0.20 & 0.60
\end{bmatrix},
&
\Trans((2,2)) &= \begin{bmatrix}
0.10 & 0.20 & 0.70\\
0.20 & 0.10 & 0.70\\
0.40 & 0.20 & 0.40
\end{bmatrix}.
\end{align*}

For the approximate game $\approxgame$, we have
\begin{gather*}
\approxrew(1) = \begin{bmatrix}
  (0.99, 0.40) & (0.69, 1.00)\\
  (0.30, 0.99) & (0.81, 0.71)
\end{bmatrix}, 
\quad
\approxrew(2) = \begin{bmatrix}
  (0.59, 0.70) & (0.69, 0.61)\\
  (0.30, 0.80) & (0.19, 0.21)
\end{bmatrix} ,
\\
\approxrew(3) = \begin{bmatrix}
  (0.19, 0.59) & (0.09, 0.70)\\
  (0.59, 0.69) & (0.50, 0.30)
\end{bmatrix},
\end{gather*}
and
\begin{gather*}
\approxTrans((1,1)) = \begin{bmatrix}
0.45 & 0.35 & 0.20\\
0.15 & 0.45 & 0.40\\
0.45 & 0.10 & 0.45
\end{bmatrix},
\quad
\approxTrans((1,2)) = \begin{bmatrix}
0.25 & 0.45 & 0.30\\
0.25 & 0.15 & 0.60\\
0.35 & 0.30 & 0.35
\end{bmatrix},
\\
\approxTrans((2,1)) = \begin{bmatrix}
0.25 & 0.30 & 0.45\\
0.35 & 0.30 & 0.35\\
0.25 & 0.20 & 0.55
\end{bmatrix},
\quad
\approxTrans((2,2)) = \begin{bmatrix}
0.15 & 0.15 & 0.70\\
0.25 & 0.10 & 0.65\\
0.40 & 0.25 & 0.35
\end{bmatrix}.
\end{gather*}

A MPE of $\approxgame$ and the corresponding value functions (computed by solving a non-linear program as described in \cite{filar1991nonlinear}) are as follows:
\begin{align}
  \approxpol^1 &= \begin{bmatrix} 
    0.33 & 0.67 \\
    1.00 & 0.00 \\
    0.00 & 1.00
  \end{bmatrix},
  &
  \approxpol^2 &= \begin{bmatrix} 
    0.13 & 0.87 \\
    1.00 & 0.00 \\
    1.00 & 0.00
  \end{bmatrix},
  \label{eq:example-MPE}
  \\
  \approxValf^1_{\approxpol} &= \begin{bmatrix} 
    0.6327  \\
    0.6170  \\
    0.6187
  \end{bmatrix},
  &
  \approxValf^2_{\approxpol} &= \begin{bmatrix} 
    0.7258  \\
    0.7148  \\
    0.7148
  \end{bmatrix}.
\end{align}

In~\eqref{eq:example-MPE}, the strategy is described as $\approxpol^i = [ \approxpol^i(\act^i | \st) ]_{\st \in \StSp, \act^i \in \ActSp^i}$. 
For strategy $\approxpol$ in~\eqref{eq:example-MPE}, we compute the value functions $\Valf^i_{\approxpol}$ for game $\game$ as described in Proposition~\ref{prop:MPE-charac} 
and the value functions $\Valf^i_{(*, \approxpol^{-i})}$ as described in Proposition~\ref{prop:approxMPE-charac} (see Sec.~\ref{sec:proof}). These are given by
\begin{align}
  \Valf^1_{\approxpol} &= \begin{bmatrix} 
    0.6341  \\
    0.6192  \\
    0.6209
  \end{bmatrix},
  &
  \Valf^2_{\approxpol} &= \begin{bmatrix} 
    0.7252  \\
    0.7142  \\
    0.7154
  \end{bmatrix},
  \\
  \Valf^1_{(*,\approxpol^2)} &= \begin{bmatrix} 
    0.6394  \\
    0.6222  \\
    0.6241
  \end{bmatrix},
  &
  \Valf^2_{(\approxpol^1, *)} &= \begin{bmatrix} 
    0.7280  \\
    0.7158  \\
    0.7171
  \end{bmatrix}.
  \label{eq:example-BR}
\end{align}
Note that 

\begin{subequations}\label{eq:example-bound}
\begin{align}
  \alpha^1_* &= \| \Valf^1_{(*, \approxpol^2)} - \Valf^1_{\approxpol} \|_\infty  = 0.005300,
  \\
  \alpha^2_* &= \| \Valf^2_{(\approxpol^1, *)} - \Valf^2_{\approxpol} \|_\infty  = 0.002785.
\end{align}
\end{subequations}
Thus, $\approxpol$ is a $(0.005300,0.002785)$-MPE of $\game$.

Now, we compare $\alpha_*$ with the bounds that we obtain using Theorem~\ref{thm:main}. 
Note that since $\alpha_{\mathrm{opt}} = 0$, $\tilde \pol^i$ defined in Theorem~\ref{thm:main} is equal to~$\approxpol^i$. Therefore, the first upper bound on $\alpha$ is given by $2\varepsilon + 2\discount \Delta^i_{\approxpol}/(1-\discount)$.
Note that 
\[
  \max_{\act \in \ActSp} \max_{\st \in \StSp} \lvert \rew(\st, \act) - \approxrew(\st, \act) \rvert = 0.01.
\]
Thus, $\varepsilon = 0.01$. Moreover,
\begin{align*}
  \Delta^1_{\approxpol} &= 
  \max_{\st \in \St, \act \in \Act} \biggl| \sum_{\st' \in \St} 
    \biggl[ \Trans(\st' | \st, \act) \approxValf^{1}_{\approxpol}(\st')
     - \approxTrans(\st' | \st, \act) \approxValf^{1}_{\approxpol}(\st') \biggr]
    \biggr|
  =
  0.000784
  \\
  \Delta^2_{\approxpol} &= 
  \max_{\st \in \St, \act \in \Act} \biggl| \sum_{\st' \in \St} 
    \biggl[ \Trans(\st' | \st, \act) \approxValf^{2}_{\approxpol}(\st')
     - \approxTrans(\st' | \st, \act) \approxValf^{2}_{\approxpol}(\st') \biggr]
    \biggr|
    =
  0.000550
\end{align*}
Then, by Theorem~\ref{thm:main}, we have that 
\[
  \alpha \le 2\biggl( \varepsilon + \frac{\gamma \Delta_{\approxpol}}{1- \gamma} \biggr)
  =
  2 \times 0.01 + \frac{2 \times 0.9}{0.1}
  \begin{bmatrix}
        0.000784\\
        0.000550
      \end{bmatrix}
      =
      \begin{bmatrix}
        0.034112\\
        0.029900
      \end{bmatrix}.
\]

Now, we consider the upper bound on $\Delta^i_{\approxpol}$ in terms of $\rho_F(\approxValf^i_{\approxpol})$. 
\begin{enumerate}
  \item We first consider the case when $\F = \FTV$. Note that
    \[
      \max_{\act \in \ActSp} \max_{\st \in \StSp} \TV(\Trans(\cdot | \st, \act), \approxTrans(\cdot | \st, \act)) = 0.05,\\
    \]
    Thus when $\F = \FTV$, $\approxgame$ is a $(0.01, 0.05)$-approximation of game $\game$. Also note that $\Span(\approxValf^1_{\approxpol}) = 0.015684$ and $\Span(\approxValf^2_{\approxpol}) = 0.010990$. Then, from Theorem~\ref{thm:main}, we have that
    \begin{gather*}
      \alpha \le 2\biggl( \varepsilon + \frac{\gamma \delta \Span(\approxValf^i_{\approxpol})}{1- \gamma} \biggr)
      = 2 \times 0.01 + \frac{2 \times 0.9 \times 0.05 }{0.1}
      \begin{bmatrix}
        0.015684\\
        0.010990
      \end{bmatrix}
      =
      \begin{bmatrix}
        0.034116\\
        0.029903
      \end{bmatrix}.
    \end{gather*}
    
    \medskip
  \item Now we equip the state space $\StSp$ with a metric $d$ where $d(\st, \st') = \lvert \st - \st' \rvert$ and consider the case $\F = \FW$. Note that
    \[
      \max_{\act \in \ActSp} \max_{\st \in \StSp} \W(\Trans(\cdot | \st, \act), \approxTrans(\cdot | \st, \act)) = 0.10.
    \]
    Thus when $\F = \FW$, $\approxgame$ is a $(0.01, 0.10)$-approximation of game $\game$. Also note that $\Lipschitz(\approxValf^1_{\approxpol}) = 0.015684$ and $\Lipschitz(\approxValf^2_{\approxpol}) = 0.010990$. Then, from Theorem~\ref{thm:main}, we have that
    \begin{gather*}
      \alpha \le 2\biggl( \varepsilon + \frac{\gamma \delta \Lipschitz(\approxValf^i_{\approxpol})}{1- \gamma} \biggr)
      = 2 \times 0.01 + \frac{2 \times 0.9 \times 0.10 }{0.1}
      \begin{bmatrix}
        0.015684\\
        0.010990
      \end{bmatrix}
      =
      \begin{bmatrix}
        0.048231\\
        0.039782
      \end{bmatrix}.
    \end{gather*}
\end{enumerate}
The above example shows that with the given $(\varepsilon,\delta)$, the bound of Theorem~\ref{thm:main} is loose by only a small multiplicative factor of approximately 6 to 15. 

\subsection{Sample complexity of generative models}
We now consider the setting for model based MARL. Consider the game $\game$ described in Sec.~\ref{sec:example} but suppose that the transition matrix $\Trans$ is not known but we have access to a generative model which can generate samples from $\Trans$. We assume that we are interested in identifying an $\alpha$-MPE, where $\alpha = 0.1$ with probability $1-p = 0.99$. We assume that we can exactly compute the MPE of the approximated game, i.e., $\alpha_{\mathrm{opt}} = 0$. 

From Theorem~\ref{thm:MARL}, we have an upper bound on the number $n$ of samples for each state-action pair as
\begin{align*}
    n &\ge \biggl\lceil \biggl( \frac{ \discount }{(1 - \discount)^2} \Span(r) \biggr)^2 
    \frac{ 2 \log(2 |\StSp|\bigl(\prod_{i \in \playerset}|\ActSp^i| \bigr) |\playerset| p^{-1}) }
    {\alpha^2} \biggr\rceil \\
    &= \biggl\lceil \biggl( \frac{ 0.9 }{(1 - 0.9)^2}  \times 0.9 \biggr)^2 
    \frac{ 2 \log(2 \times 3 \times 4 \times 2 / 0.01) }
    {0.1^2} \biggr\rceil = \NH.
\end{align*}

Note that $\gamma = 0.9 \ge \GAMMAN$, which is the lower bound on the critical discount factor $\discount^\circ$ calculated in Lemma~\ref{lemma:c}. So, we don't know upfront whether the Hoeffding-type bound is tighter than the Bernstein-type bound. We compute the number of samples $n$ for each state-action pair based on Theorem~\ref{thm:MARL2}, which is given by  
\begin{align*}
    n &\ge \biggl\lceil \frac{c\discount \|r\|_{\infty}^2\log(8|\StSp|\bigl(\prod_{i \in \playerset}|\ActSp^i| \bigr)|\playerset|(1-\discount)^{-2}p^{-1})}{(1-\discount)\alpha^2} \biggr\rceil \\
    &\ge \biggl\lceil \frac{64 \times 0.9 \times 1.0^2 
    \log(8 \times 3 \times 4 \times 2 / ((1 - 0.9)^2 \times 0.01) )}{(1-0.9) \times 0.1^2}\biggr\rceil = \NB.
\end{align*}
which is significantly smaller than the sample complexity bound of $n \ge \NH$ calculated using Theorem~\ref{thm:MARL} above.

We now verify the Bernstein bound via simulation. We run $M = 1,000$ experiments. For each experiment, we generate $n = \NB$ samples for each state-action pair, and estimate an empirical model $\approxTrans_n(s'|s,a) = \text{count}(s'|s,a)/n$. We compute the MPE $\approxpol_n$ for the approximate game $\approxgame = \langle \StSp, \{\ActSp\}_{i \in \playerset}, \approxTrans, \rew, \discount \rangle$. Then, using Proposition~\ref{prop:approxMPE-charac}, we compute $\alpha_n = (\alpha^1_n, \alpha^2_n)$ such that $\approxpol$ is a $\alpha$-MPE of game~$\game$. The scatter plot of $\alpha_n = (\alpha^1_n, \alpha^2_n)$ along with the empirical distribution of $\alpha^1_n$ and $\alpha^2_n$ are shown in Fig.~\ref{fig:Rplot_histogram}. The values for $\alpha^1_n$ are typically larger than the values for $\alpha^2_n$ because of the parameters of the specific game $\game$ chosen in the example. This trend is also apparent in the analytical $\alpha$ calculated previously based on worst case errors. Note that for most cases, both $\alpha^1_n$ and $\alpha^2_n$ are smaller than $10^{-4}$, which is much less than our target $\alpha$ of $0.1$. This highlights the looseness of the upper bound in Theorem~\ref{thm:MARL2}.

\begin{figure}[!t]
    \centering
    \includegraphics[width=0.9\textwidth]{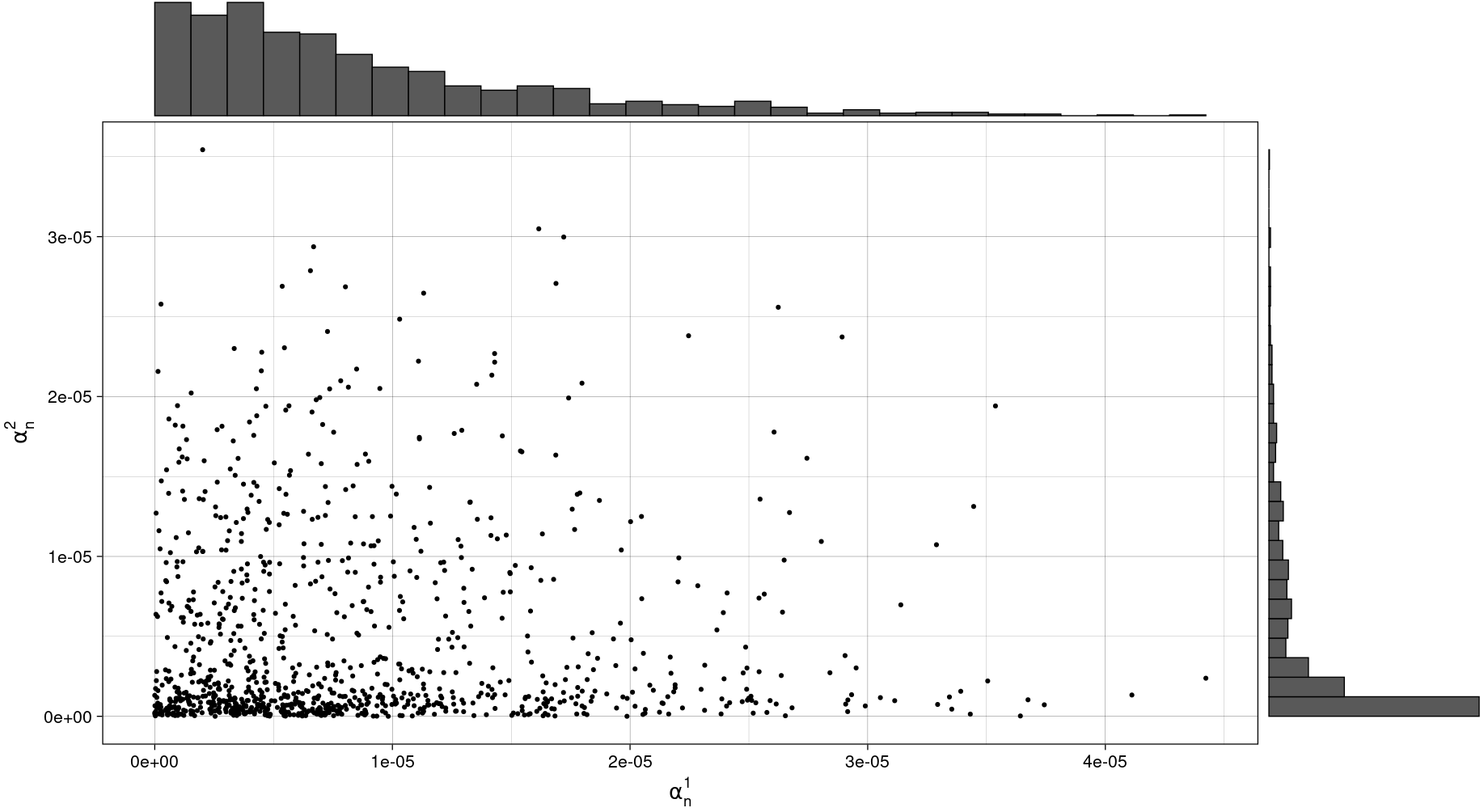}
    \caption{Scatter plot of $(\alpha^1_n,\alpha^2_n)$ such that the policy $\approxpol_n$ of game $\approxgame_n$ is a $(\alpha^1_n,\alpha^2_n)$-MPE of game $\game$ for $M=1,000$ independently generated experiments for $n=\NB$}. The histogram of the marginal probability distribution for $\alpha^1_n$ and $\alpha^2_n$ are shown on the top and the right.
    \label{fig:Rplot_histogram}
\end{figure}

\section{Conclusion}\label{sec:conclusion}
In this paper, we quantify how robust MPE are to model approximations given the degree of approximation in the approximate game. We provide bounds on the degree of approximation based on the approximation error in the reward and transition functions and properties of the value function of the MPE. We also present coarser, instance independent upper bounds, which do not depend of the value function but only depend on the properties of the reward and transition function of the approximate game. 
Using these approximation bounds, we provide sample complexity bounds for computing an approximate MPE using a generative model.

An interesting feature of the results is that the approximation bounds depend on the choice of the metric on probability spaces. We work with a class of metrics known as IPMs and specialize our results for two specific choices of IPMs: total variation distance and Wasserstein distance. However, the results are applicable to any IPM. For games with high-dimensional state spaces, metrics such as maximum mean discrepancy \citep{Sriperumbudur2008} might be more appropriate. 

The generative model setting considered in this paper circumvents the exploration problem of learning. Therefore, the sample complexity results presented in this paper should be viewed as a lower bound on the number of samples required by an algorithm to learn an $\alpha$-MPE. The proposed algorithm is not practical as it requires storing the transition matrix, which has a size of $|\StSp| \prod_{i \in \playerset} |\ActSp^i|$, which is exponential in the number of agents. The algorithm also assumes that there is a system planner which computes the approximate MPE. Developing a scalable and distributing algorithm for learning MPE remains a challenging open problem.

We conclude by noting that the results presented in this paper were restricted to Markov games with perfect information. An interesting future direction is to develop similar approximation bounds for Markov games with imperfect information as well as specific classes of dynamic games such as mean-field games and their variants.

\section*{Compliance with Ethical Standards: Conflict of Interest Statement}

The authors declare that they have no conflict of interest.

\section*{Data Availability Statement}

Data sharing not applicable to this article as no datasets were generated or analysed during the current study.

\appendix

\section{Proof of Theorem~\ref{thm:MDP-approx}}\label{app:MDP}
Let $\Valf_*$ denote the optimal value function for MDP $\MDP$ and $\Valf_{\pol}$ denote the value function for policy $\pol$ in MDP $\MDP$. Let $\approxpol_*$ be the optimal policy for $\approxMDP$ and $\approxpol$ be an $\alpha_{\mathrm{opt}}$-optimal policy for $\approxMDP$. 
From triangle inequality, we have
\begin{equation}
  \| \Valf_* - \Valf_{\approxpol} \|_{\infty} 
  \le
  \| \Valf_* - \approxValf_{\approxpol_*} \|_{\infty} 
  +
  \| \approxValf_{\approxpol_*} - \approxValf_{\approxpol} \|_{\infty}
  + 
  \| \Valf_{\approxpol} - \approxValf_{\approxpol} \|_{\infty} .
  \label{eq:triangle-ineq}
\end{equation}
Now we bound the three terms separately. For the first term, we have
\begin{align*}
  \| \Valf_* - \approxValf_{\approxpol_*} \|_{\infty} 
  &\stackrel{(a)}\le
  \max_{\st \in \StSp} \biggl|
  \max_{\act \in \ActSp} \biggl[ 
    (1-\discount)\rew(\st, \act) + \discount \sum_{\st' \in \StSp}
    \Trans(\st' | \st, \act) \Valf_*(\st') 
  \notag \\
  &\hskip 5em
  -
    (1-\discount)\approxrew(\st, \act) - \discount \sum_{\st' \in \StSp}
    \approxTrans(\st' | \st, \act) \approxValf_{\approxpol_*}(\st') 
  \biggr] \biggr|
  \displaybreak[1] \notag \\
  &\stackrel{(b)}\le
  (1-\discount) \max_{(\st,\act) \in \StSp \times \ActSp}
  \bigl| \rew(\st, \act) - \approxrew(\st, \act) \bigr|
  \notag \\
  &\quad + \discount\max_{(\st,\act) \in \StSp \times \ActSp}
  \biggl| \sum_{\st' \in \StSp}
    \Trans(\st' | \st, \act) \Valf_*(\st') 
    -
    \Trans(\st' | \st, \act) \approxValf_{\approxpol_*}(\st') 
  \biggr|
  \notag \\
  &\quad + \discount\max_{(\st,\act) \in \StSp \times \ActSp}
  \biggl| \sum_{\st' \in \StSp}
    \Trans(\st' | \st, \act) \approxValf_{\approxpol_*}(\st') 
    -
    \approxTrans(\st' | \st, \act) \approxValf_{\approxpol_*}(\st') 
  \biggr|
  \displaybreak[1] \notag \\
  &\stackrel{(c)}\le (1 - \discount) \varepsilon + \discount \| \Valf_* - \approxValf_{\approxpol_*}\|_{\infty} + \discount \Delta_{\approxpol_*},
\end{align*}
where $(a)$ relies on the fact that $\max f(x) \le \max | f(x) - g(x) | + \max g(x)$, $(b)$ follows from triangle inequality, and $(c)$ follows from the definition of $\varepsilon$ and $\Delta_{\approxpol}$. Therefore,
\begin{equation}\label{eq:triangle-ineq-1}
  \| \Valf_* - \approxValf_{\approxpol_*} \|_{\infty} 
  \le 
  \varepsilon + \frac{\discount \Delta_{\approxpol_*}}{1-\discount}.
\end{equation}
For the second term of~\eqref{eq:triangle-ineq}, we have
\begin{equation}
    \| \approxValf_{\approxpol_*} - \approxValf_{\approxpol} \|_{\infty} \le \alpha_{\mathrm{opt}},
    \label{eq:triangle-ineq-2}
\end{equation}
since $\approxpol$ is an $\alpha_{\mathrm{opt}}$-optimal policy of $\approxMDP$.

For the third term of~\eqref{eq:triangle-ineq}, we have
\begin{align*}
  \| \Valf_{\approxpol} - \approxValf_{\approxpol} \|_{\infty} 
  &=
  \max_{\st \in \StSp} \biggl|
  \sum_{\act \in \ActSp} \approxpol(\act | \st)
  \biggl[
    (1-\discount)\rew(\st, \act) + \discount \sum_{\st' \in \StSp}
    \Trans(\st' | \st, \act) \Valf_{\approxpol}(\st') 
    \notag \\
    &\hskip 8em
    -
    (1-\discount) \approxrew(\st,\act) - \discount \sum_{\st' \in \StSp}
    \approxTrans(\st' | \st, \act) \approxValf_{\approxpol}(\st') 
  \biggr] \biggl|
    \displaybreak[1] \notag \\
  &\stackrel{(d)}\le 
  (1-\discount)\max_{\st \in \StSp} \biggl|
  \sum_{\act \in \ActSp} \approxpol(\act | \st)
   \bigl[ \rew(\st,\act) - \approxrew(\st, \act) \bigr] \biggr|
  \notag \\
  &\quad + \discount
  \max_{\st \in \StSp} \biggl|
  \sum_{\act \in \ActSp} \approxpol(\act | \st)
  \biggl[ 
    \sum_{\st' \in \StSp} \bigl[
      \Trans(\st' | \st, \act) \Valf_{\approxpol}(\st')
      -
      \Trans(\st' | \st, \act) \approxValf_{\approxpol}(\st')
  \bigr] \biggr] \biggr|
  \notag \\
  &\quad + \discount
  \max_{\st \in \StSp} \biggl|
  \sum_{\act \in \ActSp} \approxpol(\act | \st)
  \biggl[ 
    \sum_{\st' \in \StSp} \bigl[
      \Trans(\st' | \st, \act) \approxValf_{\approxpol}(\st')
      -
      \approxTrans(\st' | \st, \act) \approxValf_{\approxpol}(\st')
  \bigr] \biggr] \biggr|
    \displaybreak[1] \notag \\
  &\stackrel{(e)}\le (1-\discount)\varepsilon + \discount \| \Valf_{\approxpol} - \approxValf_{\approxpol} \|_{\infty}  + \discount \Delta_{\approxpol}
\end{align*}
where $(d)$ follows from triangle inequality and $(e)$ follows from the definition of $\varepsilon$ and $\Delta_{\approxpol}$. Therefore,
\begin{equation}\label{eq:triangle-ineq-3}
  \| \Valf_{\approxpol} - \approxValf_{\approxpol} \|_{\infty}  
  \le \varepsilon + \frac{\discount \Delta_{\approxpol}}{1-\discount}.
\end{equation}
The result of Theorem~\ref{thm:MDP-approx} then follows by substituting~\eqref{eq:triangle-ineq-1}--\eqref{eq:triangle-ineq-3} in~\eqref{eq:triangle-ineq}. 
\qed

\section{Proof of Theorem~\ref{cor:MDP-approx}}\label{app:MDP2}
We follow the same notation as in Appendix~\ref{app:MDP}. In addition, let $\Bellman_*$ and $\Bellman_\pol$ denote the Bellman operators for the true model and let $\hat{\Bellman}_*$ and $\hat \Bellman_\pol$ denote the Bellman operators for the approximate model. As in Appendix~\ref{app:MDP},
the approximation error can be bounded as
\begin{equation}
  \| \Valf_* - \Valf_{\approxpol} \|_{\infty} 
  \le
  \| \Valf_* - \approxValf_{\approxpol_*} \|_{\infty} 
  +
  \| \approxValf_{\approxpol_*} - \approxValf_{\approxpol}\|_{\infty}
  +
  \| \Valf_{\approxpol} - \approxValf_{\approxpol} \|_{\infty}.
  \label{eq:triangle2}
\end{equation}
Recall that $\Valf_{\approxpol} = \Bellman_{\approxpol} \Valf_{\approxpol}$ and $\approxValf_{\approxpol} = \hat{\Bellman}_{\approxpol} \approxValf_{\approxpol}$. Moreover, $\approxValf_* = \approxValf_{\approxpol_*}$. Therefore, from triangle inequality, we have
\begin{align}
  \| \Valf_* - \Valf_{\approxpol} \|_{\infty} 
  &\le 
  \| \Valf_* - \approxValf_{*} \|_{\infty} 
  +
  \| \approxValf_{\approxpol_*} - \approxValf_{\approxpol}\|_{\infty}
  + 
  \| \Bellman_{\approxpol} \Valf_{\approxpol} - \Bellman_{\approxpol} \Valf_* \|_{\infty}
  \notag \\
  & \quad + 
  \| \Bellman_{\approxpol} \Valf_* - \hat {\Bellman}_{\approxpol} \Valf_* \|_{\infty}
  +
  \| \hat {\Bellman}_{\approxpol} \Valf_* - \hat {\Bellman}_{\approxpol} \approxValf_{*} \|_{\infty}
  +
  \| \hat {\Bellman}_{\approxpol} \approxValf_{*} - \hat {\Bellman}_{\approxpol} \approxValf_{\approxpol}  \|_{\infty}
  \notag \\
  &\stackrel{(a)}\le
  \| \Valf_* - \approxValf_{*} \|_{\infty} 
  +
  \| \approxValf_{\approxpol_*} - \approxValf_{\approxpol}\|_{\infty}
  + 
  \discount \| \Valf_{\approxpol} - \Valf_* \|_{\infty}
  \notag \\
  & \quad + 
  \| \Bellman_{\approxpol} \Valf_* - \hat {\Bellman}_{\approxpol} \Valf_* \|_{\infty}
  +
  \discount \|  \Valf_* -  \approxValf_{*} \|_{\infty}
  +
  \discount \| \approxValf_{*} - \approxValf_{\approxpol}  \|_{\infty}
\end{align}
where $(a)$ from the contraction property of the Bellman operators. Rearranging terms and using~\eqref{eq:triangle-ineq-2}, we get that
\begin{align}
  \| \Valf_* - \Valf_{\approxpol} \|_{\infty} 
  &\le
  \frac{1}{1-\discount} \bigl[
    (1+\discount)  \| \Valf_* - \approxValf_{*} \|_{\infty} 
    +
    (1 + \discount) \alpha_{\mathrm{opt}}
    + 
  \| \Bellman_{\approxpol} \Valf_* - \hat {\Bellman}_{\approxpol} \Valf_* \|_{\infty}
  \bigr]
  \label{eq:triangle2-1}
\end{align}
Now the first term in~\eqref{eq:triangle2-1} can be simplified similar to~\eqref{eq:triangle-ineq-1}, where in step $(b)$ of~\eqref{eq:triangle-ineq-1}, we need to add and subtract $\sum_{\st' \in \StSp} \approxTrans(\st' | \st, \act) \Valf_{\pol_*}(s')$. Using this, we would obtain
\begin{equation}\label{eq:triangle2-ineq-1}
  \| \Valf_* - \approxValf_{*} \|_{\infty} 
  \le 
  \varepsilon + \frac{\discount \overline \Delta_{\pol_*}}{1-\discount}.
\end{equation}
The last term in~\eqref{eq:triangle2-1} can be simplified as follows:
\begin{align}
  \| \Bellman_{\approxpol} \Valf_* - \hat {\Bellman}_{\approxpol} \Valf_* \|_{\infty}
  &\le
  \max_{\st \in \StSp} \sum_{\act \in \ActSp} \approxpol(\act | \st)
  \biggl[ | \rew(\st,\act) - \approxrew(\st,\act) | \notag \\
  & \hskip 7em + 
  \discount \biggl| \sum_{\st' \in \StSp} \Trans(\st'|\st,\act) \Valf_*(\st')
  - \sum_{\st' \in \StSp} \approxTrans(\st' | \st, \act) \Valf_*(\st') \biggr| \biggr]
  \notag \\
  &\le \varepsilon + \discount \overline \Delta_{\pol_*}
  \label{eq:triangle2-ineq-2}
 \end{align}
The result of Theorem~\ref{cor:MDP-approx} follows by substituting~\eqref{eq:triangle2-ineq-1} and~\eqref{eq:triangle2-ineq-2} in~\eqref{eq:triangle2-1}.
\qed


\begin{thebibliography}{65}
\providecommand{\natexlab}[1]{#1}
\providecommand{\url}[1]{\texttt{#1}}
\expandafter\ifx\csname urlstyle\endcsname\relax
  \providecommand{\doi}[1]{doi: #1}\else
  \providecommand{\doi}{doi: \begingroup \urlstyle{rm}\Url}\fi

\bibitem[Acemoglu and Robinson(2001)]{acemoglu2001theory}
D.~Acemoglu and J.~A. Robinson.
\newblock A theory of political transitions.
\newblock \emph{American Economic Review}, 91\penalty0 (4):\penalty0 938--963,
  2001.

\bibitem[Agarwal et~al.(2020)Agarwal, Kakade, and Yang]{agarwal2020model}
A.~Agarwal, S.~Kakade, and L.~F. Yang.
\newblock Model-based reinforcement learning with a generative model is minimax
  optimal.
\newblock In \emph{Conference on Learning Theory}, pages 67--83. PMLR, 2020.

\bibitem[Aguirregabiria and Mira(2007)]{aguirregabiria2007sequential}
V.~Aguirregabiria and P.~Mira.
\newblock Sequential estimation of dynamic discrete games.
\newblock \emph{Econometrica}, 75\penalty0 (1):\penalty0 1--53, 2007.

\bibitem[Akchurina(2010)]{akchurina2010multi}
N.~Akchurina.
\newblock \emph{Multi-agent reinforcement learning algorithms.}
\newblock PhD thesis, University of Paderborn, 2010.

\bibitem[Albright and Winston(1979)]{albright1979birth}
S.~C. Albright and W.~Winston.
\newblock A birth--death model of advertising and pricing.
\newblock \emph{Advances in Applied Probability}, 11\penalty0 (1):\penalty0
  134--152, 1979.

\bibitem[Altman(1999)]{altman1999constrained}
E.~Altman.
\newblock \emph{Constrained Markov decision processes: stochastic modeling}.
\newblock CRC Press, 1999.

\bibitem[Azar et~al.(2013)Azar, Munos, and Kappen]{azar2013minimax}
M.~G. Azar, R.~Munos, and H.~J. Kappen.
\newblock Minimax {PAC} bounds on the sample complexity of reinforcement
  learning with a generative model.
\newblock \emph{Machine learning}, 91\penalty0 (3):\penalty0 325--349, 2013.

\bibitem[Bajari et~al.(2007)Bajari, Benkard, and Levin]{bajari2007estimating}
P.~Bajari, C.~L. Benkard, and J.~Levin.
\newblock Estimating dynamic models of imperfect competition.
\newblock \emph{Econometrica}, 75\penalty0 (5):\penalty0 1331--1370, 2007.

\bibitem[Ba{\c{s}}ar and Bernhard(2008)]{Basar2008}
T.~Ba{\c{s}}ar and P.~Bernhard.
\newblock \emph{H-infinity optimal control and related minimax design problems:
  a dynamic game approach}.
\newblock Springer Science \& Business Media, 2008.

\bibitem[Ba{\c s}ar and Zaccour(2018)]{Basar2018}
T.~Ba{\c s}ar and G.~Zaccour, editors.
\newblock \emph{Handbook of Dynamic Game Theory}.
\newblock Springer International Publishing, 2018.

\bibitem[Bertsekas(2017)]{bertsekas_book}
D.~P. Bertsekas.
\newblock Dynamic programming and optimal control.
\newblock \emph{Belmont, MA: Athena Scientific}, 2017.

\bibitem[Breton(1991)]{Breton1991}
M.~Breton.
\newblock \emph{Algorithms for Stochastic Games}, pages 45--57.
\newblock Springer, 1991.

\bibitem[Busoniu et~al.(2008)Busoniu, Babuska, and
  De~Schutter]{busoniu2008comprehensive}
L.~Busoniu, R.~Babuska, and B.~De~Schutter.
\newblock A comprehensive survey of multiagent reinforcement learning.
\newblock \emph{IEEE Transactions on Systems, Man, and Cybernetics, Part C
  (Applications and Reviews)}, 38\penalty0 (2):\penalty0 156--172, 2008.

\bibitem[Cesa-Bianchi and Lugosi(2006)]{cesa2006prediction}
N.~Cesa-Bianchi and G.~Lugosi.
\newblock \emph{Prediction, learning, and games}.
\newblock Cambridge university press, 2006.

\bibitem[Deng et~al.(2021)Deng, Li, Mguni, Wang, and Yang]{deng2021complexity}
X.~Deng, Y.~Li, D.~H. Mguni, J.~Wang, and Y.~Yang.
\newblock On the complexity of computing markov perfect equilibrium in
  general-sum stochastic games.
\newblock \emph{arXiv preprint arXiv:2109.01795}, 2021.

\bibitem[Doraszelski and Escobar(2010)]{Doraszelski2010}
U.~Doraszelski and J.~F. Escobar.
\newblock A theory of regular {Markov} perfect equilibria in dynamic stochastic
  games: {Genericity}, stability, and purification.
\newblock \emph{Theoretical Economics}, 5\penalty0 (3):\penalty0 369--402,
  2010.
\newblock ISSN 1555-7561.

\bibitem[Ericson and Pakes(1995)]{ericson1995markov}
R.~Ericson and A.~Pakes.
\newblock Markov-perfect industry dynamics: A framework for empirical work.
\newblock \emph{The Review of economic studies}, 62\penalty0 (1):\penalty0
  53--82, 1995.

\bibitem[Fershtiman and Pakes(2000)]{fershtiman2000dynamic}
C.~Fershtiman and A.~Pakes.
\newblock A dynamic oligopoly with collusion and price wars.
\newblock \emph{The RAND Journal of Economics}, 31\penalty0 (2):\penalty0
  207--236, 2000.

\bibitem[Filar and Vrieze(1996)]{FilarCompetitive1996}
J.~Filar and K.~Vrieze.
\newblock \emph{Competitive {Markov} {Decision} {Processes}}.
\newblock Springer, New York, NY, 1996.
\newblock ISBN 978-1-4612-8481-9 978-1-4612-4054-9.

\bibitem[Filar et~al.(1991)Filar, Schultz, Thuijsman, and
  Vrieze]{filar1991nonlinear}
J.~A. Filar, T.~A. Schultz, F.~Thuijsman, and O.~Vrieze.
\newblock Nonlinear programming and stationary equilibria in stochastic games.
\newblock \emph{Mathematical Programming}, 50\penalty0 (1):\penalty0 227--237,
  1991.

\bibitem[Fink(1964)]{fink1964equilibrium}
A.~M. Fink.
\newblock Equilibrium in a stochastic \$n\$-person game.
\newblock \emph{Hiroshima Mathematical Journal}, 28\penalty0 (1), 1964.

\bibitem[Herings and Peeters(2010)]{herings2010homotopy}
P.~J.-J. Herings and R.~Peeters.
\newblock Homotopy methods to compute equilibria in game theory.
\newblock \emph{Economic Theory}, 42\penalty0 (1):\penalty0 119--156, 2010.

\bibitem[Herings et~al.(2004)Herings, Peeters, et~al.]{herings2004stationary}
P.~J.-J. Herings, R.~J. Peeters, et~al.
\newblock Stationary equilibria in stochastic games: Structure, selection, and
  computation.
\newblock \emph{Journal of Economic Theory}, 118\penalty0 (1):\penalty0 32--60,
  2004.

\bibitem[Hinderer(2005)]{Hinderer2005}
K.~Hinderer.
\newblock Lipschitz continuity of value functions in {M}arkovian decision
  processes.
\newblock \emph{Mathematical Methods of Operations Research}, 62\penalty0
  (1):\penalty0 3--22, Sep 2005.
\newblock ISSN 1432-5217.

\bibitem[Hoffman and Karp(1966)]{hoffman1966nonterminating}
A.~J. Hoffman and R.~M. Karp.
\newblock On nonterminating stochastic games.
\newblock \emph{Management Science}, 12\penalty0 (5):\penalty0 359--370, 1966.

\bibitem[Jaśkiewicz and Nowak(2014)]{Jaskiewicz2014}
A.~Jaśkiewicz and A.~S. Nowak.
\newblock Robust {Markov} perfect equilibria.
\newblock \emph{Journal of Mathematical Analysis and Applications},
  419\penalty0 (2):\penalty0 1322--1332, 2014.

\bibitem[Kakade(2003)]{Kakade:PhD}
S.~M. Kakade.
\newblock \emph{On the sample complexity of reinforcement learning}.
\newblock PhD thesis, University College, London, 2003.

\bibitem[Kearns and Singh(1999)]{kearns1999finite}
M.~Kearns and S.~Singh.
\newblock Finite-sample convergence rates for q-learning and indirect
  algorithms.
\newblock In \emph{Advances in neural information processing systems}, pages
  996--1002, 1999.

\bibitem[Krupnik et~al.(2019)Krupnik, Mordatch, and
  Tamar]{Krupnik2019MultiAgent}
O.~Krupnik, I.~Mordatch, and A.~Tamar.
\newblock Multi-agent reinforcement learning with multi-step generative models.
\newblock \emph{arXiv preprint arXiv:1901.10251}, Nov. 2019.

\bibitem[Leonardos et~al.(2021)Leonardos, Overman, Panageas, and
  Piliouras]{leonardos2021global}
S.~Leonardos, W.~Overman, I.~Panageas, and G.~Piliouras.
\newblock Global convergence of multi-agent policy gradient in markov potential
  games.
\newblock \emph{arXiv preprint arXiv:2106.01969}, 2021.

\bibitem[Li et~al.(2020)Li, Wei, Chi, Gu, and Chen]{li2020breaking}
G.~Li, Y.~Wei, Y.~Chi, Y.~Gu, and Y.~Chen.
\newblock Breaking the sample size barrier in model-based reinforcement
  learning with a generative model.
\newblock \emph{Advances in neural information processing systems}, 33, 2020.

\bibitem[Littman(1994)]{littman1994markov}
M.~L. Littman.
\newblock Markov games as a framework for multi-agent reinforcement learning.
\newblock In \emph{International Conference on Machine Learning}, pages
  157--163. Elsevier, 1994.

\bibitem[Littman(2001)]{littman2001value}
M.~L. Littman.
\newblock Value-function reinforcement learning in {Markov} games.
\newblock \emph{Cognitive systems research}, 2\penalty0 (1):\penalty0 55--66,
  2001.

\bibitem[Mailath and Samuelson(2006)]{MailathRepeated}
G.~J. Mailath and L.~Samuelson.
\newblock \emph{Repeated {Games} and {Reputations}: {Long}-{Run}
  {Relationships}}.
\newblock Oxford University Press, 2006.

\bibitem[Maskin and Tirole(1988{\natexlab{a}})]{maskin1988theory}
E.~Maskin and J.~Tirole.
\newblock A theory of dynamic oligopoly, {I}: Overview and quantity competition
  with large fixed costs.
\newblock \emph{Econometrica: Journal of the Econometric Society}, pages
  549--569, 1988{\natexlab{a}}.

\bibitem[Maskin and Tirole(1988{\natexlab{b}})]{maskin1988theory2}
E.~Maskin and J.~Tirole.
\newblock A theory of dynamic oligopoly, {II}: Price competition, kinked demand
  curves, and edgeworth cycles.
\newblock \emph{Econometrica: Journal of the Econometric Society}, pages
  571--599, 1988{\natexlab{b}}.

\bibitem[Maskin and Tirole(2001)]{maskin2001markov}
E.~Maskin and J.~Tirole.
\newblock Markov perfect equilibrium: I. observable actions.
\newblock \emph{Journal of Economic Theory}, 100\penalty0 (2):\penalty0
  191--219, 2001.

\bibitem[M\"uller(1997)]{Muller1997}
A.~M\"uller.
\newblock How does the value function of a {M}arkov decision process depend on
  the transition probabilities?
\newblock \emph{Mathematics of Operations Research}, 22\penalty0 (4):\penalty0
  872--885, 1997.

\bibitem[M{\"u}ller(1997)]{muller1997integral}
A.~M{\"u}ller.
\newblock Integral probability metrics and their generating classes of
  functions.
\newblock \emph{Advances in Applied Probability}, 29\penalty0 (2):\penalty0
  429--443, 1997.

\bibitem[Pakes et~al.(2007)Pakes, Ostrovsky, and Berry]{pakes2007simple}
A.~Pakes, M.~Ostrovsky, and S.~Berry.
\newblock Simple estimators for the parameters of discrete dynamic games (with
  entry/exit examples).
\newblock \emph{the RAND Journal of Economics}, 38\penalty0 (2):\penalty0
  373--399, 2007.

\bibitem[P{\'e}rolat et~al.(2017)P{\'e}rolat, Strub, Piot, and
  Pietquin]{perolat2017learning}
J.~P{\'e}rolat, F.~Strub, B.~Piot, and O.~Pietquin.
\newblock Learning {Nash} equilibrium for general-sum {Markov} games from batch
  data.
\newblock In \emph{Artificial Intelligence and Statistics}, pages 232--241.
  PMLR, 2017.

\bibitem[Pesendorfer and Schmidt-Dengler(2008)]{pesendorfer2008asymptotic}
M.~Pesendorfer and P.~Schmidt-Dengler.
\newblock Asymptotic least squares estimators for dynamic games.
\newblock \emph{The Review of Economic Studies}, 75\penalty0 (3):\penalty0
  901--928, 2008.

\bibitem[Prasad et~al.(2015)Prasad, LA, and Bhatnagar]{prasad2015two}
H.~Prasad, P.~LA, and S.~Bhatnagar.
\newblock Two-timescale algorithms for learning {Nash} equilibria in
  general-sum stochastic games.
\newblock In \emph{Proceedings of the 2015 International Conference on
  Autonomous Agents and Multiagent Systems}, pages 1371--1379, 2015.

\bibitem[Rogers(1969)]{RogersNonzerosum1969}
P.~D. Rogers.
\newblock \emph{Nonzero-sum stochastic games}.
\newblock PhD thesis, University of California, Berkeley, 1969.

\bibitem[Sengupta et~al.(2019)Sengupta, Chowdhary, Huang, and
  Kambhampati]{sengupta2019general}
S.~Sengupta, A.~Chowdhary, D.~Huang, and S.~Kambhampati.
\newblock General sum markov games for strategic detection of advanced
  persistent threats using moving target defense in cloud networks.
\newblock In \emph{International Conference on Decision and Game Theory for
  Security}, pages 492--512. Springer, 2019.

\bibitem[Shapley(1953)]{Shapley_1953}
L.~S. Shapley.
\newblock Stochastic games.
\newblock \emph{Proceedings of the National Academy of Sciences}, 39\penalty0
  (10):\penalty0 1095--1100, 1953.

\bibitem[Shoham et~al.(2003)Shoham, Powers, and Grenager]{shoham2003multi}
Y.~Shoham, R.~Powers, and T.~Grenager.
\newblock Multi-agent reinforcement learning: a critical survey.
\newblock Technical report, Stanford University, 2003.

\bibitem[Sidford et~al.(2018)Sidford, Wang, Wu, Yang, and Ye]{sidford2018near}
A.~Sidford, M.~Wang, X.~Wu, L.~F. Yang, and Y.~Ye.
\newblock Near-optimal time and sample complexities for solving markov decision
  processes with a generative model.
\newblock In \emph{Proceedings of the 32nd International Conference on Neural
  Information Processing Systems}, pages 5192--5202, 2018.

\bibitem[Sidford et~al.(2020)Sidford, Wang, Yang, and Ye]{sidford2020solving}
A.~Sidford, M.~Wang, L.~Yang, and Y.~Ye.
\newblock Solving discounted stochastic two-player games with near-optimal time
  and sample complexity.
\newblock In \emph{International Conference on Artificial Intelligence and
  Statistics}, pages 2992--3002. PMLR, 2020.

\bibitem[Solan(2021)]{Solan2021}
E.~Solan.
\newblock \emph{A Course in Stochastic Game Theory}.
\newblock Cambridge University Press, 2021.

\bibitem[Song et~al.(2021)Song, Mei, and Bai]{song2021can}
Z.~Song, S.~Mei, and Y.~Bai.
\newblock When can we learn general-sum markov games with a large number of
  players sample-efficiently?
\newblock \emph{arXiv preprint arXiv:2110.04184}, 2021.

\bibitem[Sriperumbudur et~al.(2008)Sriperumbudur, Gretton, Fukumizu, Lanckriet,
  and Sch\"{o}lkopf]{Sriperumbudur2008}
B.~K. Sriperumbudur, A.~Gretton, K.~Fukumizu, G.~R.~G. Lanckriet, and
  B.~Sch\"{o}lkopf.
\newblock Injective {H}ilbert space embeddings of probability measures.
\newblock In \emph{Conference on Learning Theory}, 2008.

\bibitem[Subramanian et~al.(2022)Subramanian, Sinha, Seraj, and
  Mahajan]{subramanian2020approximate}
J.~Subramanian, A.~Sinha, R.~Seraj, and A.~Mahajan.
\newblock Approximate information state for approximate planning and
  reinforcement learning in partially observed systems.
\newblock \emph{J. Mach. Learn. Res.}, 23:\penalty0 12--1, 2022.

\bibitem[Sutton(1990)]{sutton1990}
R.~S. Sutton.
\newblock Integrated architectures for learning, planning, and reacting based
  on approximating dynamic programming.
\newblock In \emph{International Conference on Machine Learning}, pages
  216--224. San Francisco (CA), 1990.

\bibitem[Takahashi(1964)]{TakahashiEquilibrium1964}
M.~Takahashi.
\newblock Equilibrium points of stochastic non-cooperative {$n$}-person games.
\newblock \emph{Hiroshima Mathematical Journal}, 28\penalty0 (1), Jan. 1964.

\bibitem[Tidball and Altman(1996)]{Tidball1996Approximations}
M.~M. Tidball and E.~Altman.
\newblock Approximations in dynamic zero-sum games {I}.
\newblock \emph{SIAM Journal on Control and Optimization}, 34\penalty0
  (1):\penalty0 311--328, Jan. 1996.

\bibitem[Tidball et~al.(1997)Tidball, Pourtallier, and
  Altman]{Tidball1997Approximations}
M.~M. Tidball, O.~Pourtallier, and E.~Altman.
\newblock Approximations in dynamic zero-sum games {II}.
\newblock \emph{SIAM Journal on Control and Optimization}, 35\penalty0
  (6):\penalty0 2101--2117, Nov. 1997.

\bibitem[Vrieze(1987)]{VriezeStochastic1987}
O.~J. Vrieze.
\newblock \emph{Stochastic games with finite state and action spaces}.
\newblock CWI, Jan. 1987.
\newblock ISBN 978-90-6196-313-4.

\bibitem[Wang et~al.(2019)Wang, Bao, Clavera, Hoang, Wen, Langlois, Zhang,
  Zhang, Abbeel, and Ba]{Wang2019Benchmarking}
T.~Wang, X.~Bao, I.~Clavera, J.~Hoang, Y.~Wen, E.~Langlois, S.~Zhang, G.~Zhang,
  P.~Abbeel, and J.~Ba.
\newblock Benchmarking {Model}-{Based} {Reinforcement} {Learning}.
\newblock \emph{arXiv preprint arXiv:1907.02057}, July 2019.

\bibitem[Whitt(1980)]{whitt1980representation}
W.~Whitt.
\newblock Representation and approximation of noncooperative sequential games.
\newblock \emph{SIAM Journal on Control and Optimization}, 18\penalty0
  (1):\penalty0 33--48, 1980.

\bibitem[Zhang et~al.(2020)Zhang, Kakade, Basar, and Yang]{zhang2020model}
K.~Zhang, S.~Kakade, T.~Basar, and L.~Yang.
\newblock Model-based multi-agent rl in zero-sum {Markov} games with
  near-optimal sample complexity.
\newblock \emph{Advances in Neural Information Processing Systems}, 33, 2020.

\bibitem[Zhang et~al.(2021{\natexlab{a}})Zhang, Yang, and
  Ba{\c{s}}ar]{zhang2021multi}
K.~Zhang, Z.~Yang, and T.~Ba{\c{s}}ar.
\newblock Multi-agent reinforcement learning: A selective overview of theories
  and algorithms.
\newblock \emph{Handbook of Reinforcement Learning and Control}, pages
  321--384, 2021{\natexlab{a}}.

\bibitem[Zhang et~al.(2021{\natexlab{b}})Zhang, Ren, and Li]{zhang2021gradient}
R.~Zhang, Z.~Ren, and N.~Li.
\newblock Gradient play in multi-agent markov stochastic games: Stationary
  points and convergence.
\newblock \emph{arXiv e-prints}, pages arXiv--2106, 2021{\natexlab{b}}.

\bibitem[Zhang et~al.(2021{\natexlab{c}})Zhang, Wang, Shen, and
  Zhou]{Zhang2021ModelBased}
W.~Zhang, X.~Wang, J.~Shen, and M.~Zhou.
\newblock Model-based {Multi}-agent {Policy} {Optimization} with {Adaptive}
  {Opponent}-wise {Rollouts}.
\newblock In \emph{International Joint Conference on Artificial Intelligence}.
  Montreal, Canada, 2021{\natexlab{c}}.

\bibitem[Zinkevich et~al.(2006)Zinkevich, Greenwald, and
  Littman]{zinkevich2006cyclic}
M.~Zinkevich, A.~Greenwald, and M.~Littman.
\newblock Cyclic equilibria in {Markov} games.
\newblock In \emph{Neural Information Processing Systems}, pages 1641--1648,
  2006.

\end{thebibliography}
\end{document}